\documentclass[a4paper,UKenglish]{lipics-v2021}

\usepackage{multirow}
\usepackage{hyperref}
\usepackage{xspace}
\usepackage{amsmath,amssymb}
\usepackage{stmaryrd} 
\nolinenumbers

\usepackage{apxproof} 
\usepackage{algorithm}
\usepackage{algpseudocode}

\usepackage{todonotes} 

\usepackage{tikz-cd}

\usepackage{tikz}
 \usetikzlibrary{trees}
 \usetikzlibrary{shapes}
 \usetikzlibrary{fit}
 \usetikzlibrary{shadows}
 \usetikzlibrary{backgrounds}
 \usetikzlibrary{arrows,automata}

\tikzset{
   n/.style= {circle,fill,inner sep=1.5pt,node distance=2cm}
  ,acc/.style={circle,draw,inner sep=3pt,node distance=2cm}
  ,phantom/.style={circle},
  ,arr/.style={->, >=stealth, semithick, shorten <= 3pt, shorten >= 3pt}
}

\newcommand{\seq}{\mathsf{seq}}
\newcommand{\lft}{\mathsf{head}}

\newcommand{\takeout}[1]{\empty}

\title{Faster and Smaller Solutions of Obliging Games}

\author{Daniel Hausmann}{University of Gothenburg, Sweden}{d.hausmann@liverpool.co.uk}{https://orcid.org/0000-0002-0935-8602}{Supported by the EPSRC through grant EP/Z003121/1.}
\author{Nir Piterman}{University of Gothenburg,
	Sweden}{nir.piterman@chalmers.se}{https://orcid.org/0000-0002-8242-5357}{}

\authorrunning{D.~Hausmann and N.~Piterman} 

\Copyright{Daniel Hausmann and Nir Piterman} 

\ccsdesc[100]{\textcolor{red}{Theory of computation - Formal languages
		and automata theory}}

\keywords{Two-player games, reactive synthesis, Emerson-Lei games, parity games} 

\category{} 

\relatedversion{} 

\funding{Both authors supported by the ERC Consolidator grant D-SynMA (No. 772459).}

\begin{document}

\maketitle
\begin{abstract}
Obliging games have been introduced in the context of the game perspective on reactive synthesis in order to enforce a degree of cooperation between the to-be-synthesized system and the environment. Previous approaches to the analysis of obliging games have been small-step in the sense that they have been based on a reduction to standard (non-obliging) games in which single moves correspond to single moves in the original (obliging) game. Here, we propose a novel, large-step view on obliging games, reducing them to standard games in which single moves encode long-term behaviors in the original game. This not only allows us to give a meaningful definition of the environment winning in obliging games, but also leads to significantly improved bounds on both strategy sizes and the solution runtime for obliging games.
\end{abstract}

\section{Intro}

Infinite duration games~\cite{GraedelThomasWilke02} and their analysis are central to various logical problems
in computer science; problems with existing game reductions subsume model checking~\cite{BradfieldW18,HausmannS19} and satisfiability checking~\cite{FriedmannLange13a,FriedmannLatteLange13}
for temporal logics (such as CTL or the $\mu$-calculus), or reactive synthesis for LTL specifications~\cite{EsparzaKRS22}. Game arenas that are used in such reductions typically incorporate two antagonistic players (called player $\exists$ and player $\forall$ in the current work) that have dual objectives. Then the reductions construct game arenas and objectives in such a way that
the instance of the original problem has a positive answer if and only if player $\exists$ has a strategy that ensures that the player's objective is satisfied (that is, player $\exists$ \emph{wins} the game). \emph{Solving} games then amounts to determining their winner.
Games with Borel objectives are known to be \emph{determined}~\cite{Martin75}, that is, for each node in them, exactly one of the players $i$ has a winning strategy of type $V^* V_i\to V$, where $V$ is the set of game nodes, and $V_i$ the set of nodes controlled by player $i$.

\emph{Reactive synthesis}~\cite{PnueliR89} considers a setting in which a system works within an antagonistic environment, and the system-enviroment interaction is modelled by means of input variables from a set $I$ (controlled by the environment), and output variables from a set $O$ (controlled by the system). The synthesis problem then takes a logical specification $\varphi$ over the variables $I\cup O$ as input and asks for the automated construction of a system in which all interactions between the system and the environment satisfy the specification;
if such a system exists, then $\varphi$ is said to be \emph{realizable}. The reactive
synthesis problem therefore goes beyond checking realizability by also asking for a witnessing system
in the case that the input specification is realizable.
While the problem has been shown to be \textsc{2ExpTime}-complete for specifications that
are formulated in LTL, a landscape of algorithms and implementations has been developed that shows good performance
in (some) practical use cases (e.g.~\cite{LuttenbergerMS20,DijkAT24}). The feasability of these algorithms is largely owed to an  underlying reduction to infinite-duration two-player games, most commonly with parity objectives. Answering the realizability problem then corresponds to deciding the winner of the reduced game, while the construction of a witnessing system for a realizable specification corresponds to the extraction of a winning strategy from that game.
This motivates the
interest not only in game solving algorithms, but also in the analysis and extraction of winning strategies. In particular, the amount of memory required by winning strategies in
the types of games at hand determines the minimum size of witnessing systems.

Building on the described game perspective on reactive
synthesis, \emph{obliging games} have been proposed to deal
with the situation that a system might trivially realize a
specification by disallowing most or all interactions with the
environment (cf.~\cite{BloemEJK14,ChatterjeeHL10}). Obliging games address
this problem by requiring the system player to have a strategy
that not only always guarantuees that the system's strong objective (say
$\alpha_S$) is achieved, but to also always keep an interaction possible in
which intuitively both players cooperate to achieve a second, weak,
objective (say $\alpha_W$). Such a strategy then is called
\emph{graciously winning} for the system player.
For example, the generalized reactivity (GR[1]) setting (cf.~\cite{BloemJPPS12}) incorporates $k$ different
\emph{requests} and
$k$ corresponding \emph{grants} ($R_i,G_i$ for $1\leq i\leq k$).
Then the objective $\alpha_S$ states that `if all $R_i$ (requests) hold infinitely often,
then all $G_i$ (grants) hold infinitely often'. In this
setting, player
$\exists$ may satisfy the objective $\alpha_S$ trivially by simply ensuring that, for each interaction, there is some $R_i$ that is eventually avoided forever. The obliging game
setting allows to address this situation by taking $\alpha_W$ to
require that all $R_i$ hold infinitely often. Then a gracious strategy for
player $\exists$ has to ensure $\alpha_S$ (possibly by avoiding, on most interactions, some $R_i$), but it also has to enable player $\forall$ to additionally
realize $\alpha_W$, thereby always enabling at least one ``interesting'' interaction on
which all $R_i$ and also all $G_i$ hold infinitely often.

Previous approaches to the analysis of obliging games~\cite{ChatterjeeHL10,MajumdarSchmuck23} have been largely based on a reduction to equivalent non-obliging games in which the players still take single steps on the original
game graph, but in addition to that, game nodes are annotated with
auxiliary memory that is used to deal with the more involved obliging
game's objectives.
Obliging GR[1] games, as in the example above, have been considered
under the names of ``cooperative'' \cite{BloemEK15} and
``environmentally-friendly'' \cite{MajumdarPS19}.
Independent bespoke symbolic algorithms of slightly better
complexity than the general solution have been suggested.

In this context, the contributions of the current work are threefold:
\begin{itemize}
\item We propose a novel perspective on obliging games that is based on considering multi-step strategies of players in the original game, rather than emulating single-step moves.
In more detail, we provide an alternative reduction that takes obliging games to equivalent non-obliging
games in which the system player commits to certain long-term behaviors, encoded by so-called witnesses, which are just plays of the original game (we therefore call the resulting games \emph{witness games}). The environment player in turn
can either check whether a given witness is valid, that is, satisfies both $\alpha_S$ and $\alpha_W$, or accept the witness and exit it
at any game node that occurs in the witness, thereby intuitively challenging the system player to still win when a play only traverses the witness up to the exit node and then continues on outside of the proposed witness; in the latter case, the system player has to provide a new witness for the challenged game node, and so on. We use the reduction to witness games to show determinacy of obliging games with Borel objectives (however, with respect to strategies of type $V^*\to V^\omega$ for the
system player, and strategies of type $V^\omega\to V\cup V^\omega$ for
the environment player).
\item We show that witnesses for obliging games with  $\omega$-regular objectives
$\alpha_S$ and $\alpha_W$ have finite representations, as they correspond to (accepting) runs of
$\omega$-automata with acceptance condition $\alpha_S\land\alpha_W$; we call such representations \emph{certificates}.
Using certificates in place of infinite witnesses, we modify witness games to obtain
an alternative reduction that takes $\omega$-regular obliging games
to \emph{finite} $\omega$-regular non-obliging games, called \emph{certificate games}.
Technically, we present the reduction for obliging games with Emerson-Lei objectives.
 During the correctness proof for this reduction, we show that
 the memory requirements of graciously winning strategies for
 Emerson-Lei obliging games depend only linearly on the size of
 $\alpha_W$, improving significantly upon existing upper
 bounds~\cite{ChatterjeeHL10,MajumdarSchmuck23} that are, in general, exponential in the
 size of $\alpha_W$.
\item The certificate games that we propose contain an explicit game node for any possible certificate within an obliging game. Hence they are prohibitively large and it is not viable to solve them naively. However, we show how a technique of fixpoint acceleration can be used to speed up the solving process of certificate games, replacing the exploration of all certificate nodes with emptiness checking for suitable $\omega$-automata; this technique solves certificate games by computing nested fixpoints of a function that checks for the existence of suitable certificates. Thereby we are able to improve previous upper runtime bounds for the solution of obliging games; in many cases, our algorithm outperforms existing algorithms by an exponential factor.
\end{itemize}
Summing up, we propose a novel approach to the analysis of obliging games that provides more insight in their determinacy, and for
obliging games with Emerson-Lei objectives, we significantly improve existing upper bounds both on strategy sizes and solution time.\medskip

\emph{Structure.} We introduce obliging games and related notions
in Section~\ref{sec:prelim}. In Section~\ref{sec:determinacy},
we reduce obliging games to witness games and use the
reduction to show that
obliging games are determined (for strategies with additional information). Subsequently, we restrict our attention to
$\omega$-regular obliging games with Emerson-Lei objectives. In Section~\ref{sec:certgames} we reduce witness games with such objectives to
certificate games and use the latter to obtain improved upper bounds
on strategy sizes in obliging games.
In Section~\ref{sec:solving} we
show how certificate games can be solved efficiently, in consequence
improving previous upper bounds on the runtime complexity of solving
obliging games. Full proofs and additional details can be found in the appendix.

\section{Preliminaries}\label{sec:prelim}

We start by recalling obliging games
and extend the setup from previous work
to use general Borel objectives in place
of Muller objectives; we also introduce
the special case of obliging games with Emerson-Lei objectives,
and recall the definition of Emerson-Lei automata.\medskip

\noindent
\emph{Arenas, plays, strategies.} An \emph{arena} is a graph $A=(V,V_\exists,E)$, consisting of a set $V$ of \emph{nodes} and a set $E\subseteq V\times V$ of \emph{moves}; furthermore, we assume that
the set of nodes is partitioned
into the sets $V_\exists$ and $V_\forall:=V\setminus V_\exists$
of nodes \emph{owned} by player $\exists$ and
by player $\forall$, respectively.
We write $E(v)=\{w \in V\mid (v,w)\in E\}$ for the set of nodes reachable from
node $v\in V$ by a single move. We assume without loss of generality that
every node has at least one outgoing edge, that is, that $E(v)\neq\emptyset$
for all $v\in V$.
A \emph{play} $\pi=v_0 v_1\ldots$
on $A$ is a (finite or infinite) sequence of nodes such that
$v_{i+1}\in E(v_i)$ for all $i\geq 0$. The length $|\pi|=n+1$ of a finite play
$\pi=v_0 v_1\ldots v_n$ is the number of vertices it contains;
throughout, we denote the set $\{0,\ldots,n\}$ for $n\in\mathbb{N}$ by $[n]$.
By abuse of notation, we denote by $A^\omega$ the set of infinite plays
on $A$ and by $A^*$ and $A^+$ the set of finite (nonempty) plays on $A$.
A \emph{strategy} for player $i\in\{\exists,\forall\}$
is a function $\sigma:A^*\cdot V_i\to V$ that assigns a node
$\sigma(\pi v)\in V$ to every finite play $\pi v$ that ends in a node
$v\in V_i$.
A strategy $\sigma$ for player $i$ is said to be \emph{positional} if the moves that it prescribes do not depend on previously visited game nodes. Formally this is the case if we have
$\sigma(\pi v)=\sigma(\pi' v)$ for all $v\in V_i$ and all $\pi,\pi'\in A^*$.
A play $\pi=v_0 v_1\ldots$ is \emph{compatible} with a strategy $\sigma$ for
player $i\in\{\exists,\forall\}$ if
for all $j\geq 0$ such that $v_j\in V_i$, we have $v_{j+1}=\sigma(v_0 v_1\ldots v_j)$.\medskip

\noindent
\emph{Objectives and games.}
In this work we consider two types of objectives: \emph{Borel
objectives} and \emph{Emerson-Lei objectives}.
Borel objectives are explicit sets of infinite sequences of vertices.
A set is \emph{Borel definable} if it is in the $\sigma$-algebra over the open
subsets of infinite sequences of vertices $V^\omega$.
That is, sets that can be obtained by countable unions, countable
intersections, and complementations from the open sets (sets of the
form $w V^\omega$ for $w\in V^*$).
A play $\pi$ on $A$ \emph{satisfies} a Borel objective $B$ if $\pi\in B$.

Emerson-Lei objectives are
specified relative to a coloring function $\gamma_C:E\to 2^C$ (for some set
$C$ of colors) that assigns a set
$\gamma_C(v,w)\subseteq C$ of colors to every move $(v,w)\in E$; we note that our setup is more succinct than the one from~\cite{ChatterjeeHL10} where each edge has (at most) one color. A play $\pi=v_0v_1\ldots$ then induces a
sequence
$\gamma_C(\pi)=\gamma_C(v_0,v_1)\gamma_C(v_1,v_2)\ldots$ of sets of colors.
Emerson-Lei objectives are given as positive Boolean formulas
$\varphi_C\in\mathbb{B}_+(\{\mathsf{Inf}~c,\mathsf{Fin}~c\mid c\in C\})$
over atoms of the shape $\mathsf{Inf}~c$ and
$\mathsf{Fin}~c$.
Such formulas are interpreted over infinite sequences
$\gamma_0\gamma_1\ldots$ of sets of colors. We
put $\gamma_0\gamma_1\ldots\models \mathsf{Inf}~c$ if and only if there are infinitely many positions $i$ such that $c\in\gamma_i$,
and $\gamma_0\gamma_1\ldots\models \mathsf{Fin}~c$ if there are only finitely many such positions; satisfaction of Boolean operators is defined in the usual way.
Then an infinite play $\pi$ on $A$ \emph{satisfies} the formula $\varphi_C$ if and only if
$\gamma_C(\pi)\models \varphi_C$ and we define the Emerson-Lei objective induced by $\gamma_C$
and $\varphi_C$ by putting
\begin{align*}
\alpha_{\gamma_C,\varphi_C}=\{\pi\in A^\omega\mid\gamma_C(\pi)\models\varphi_C\}.
\end{align*}
\emph{Parity objectives} are a special case of Emerson-Lei objectives
with set $C=\{p_0,\ldots, p_k\}$ of colors, where each edge has exactly one color (also called \emph{priority}), and where $\varphi=\bigvee_{i \text{ even}}
\mathsf{Inf}~p_i\wedge \bigwedge_{i<j \leq k} \mathsf{Fin}~p_j$, stating
that the maximal priority that is visited infinitely often has an even index.
We can denote such objectives by just a single function $\Omega:E\to [k]$.
Further standardly used conditions include \emph{B\"uchi}, \emph{generalized B\"uchi}, \emph{generalized reactivity (GR[1])}, \emph{Rabin} and \emph{Streett} objectives, all of which are special cases of Emerson-Lei objectives (see e.g.~\cite{HausmannLP24}); the memory required for winning in such games has been investigated in~\cite{DziembowskiJW97}.

We note that neither Borel nor Emerson-Lei objectives enable a simple
distinction
between finite plays ending in existential and universal nodes; hence we will
avoid deadlocks in our game reductions, ensuring that all games in this
work allow only
infinite plays.

A \emph{standard game} is a tuple $(A,\alpha)$, where
$A=(V,V_\exists,E)$ is an arena and $\alpha$ is an objective.
A strategy $\sigma$ is \emph{winning} for player $\exists$ at some node $v\in V$
if all plays that start at $v$ and are compatible with $\sigma$ satisfy the
objective $\alpha$.
A strategy $\tau$ for player $\forall$ is defined winning dually.

An \emph{obliging game} is a tuple $(A,\alpha_S,\alpha_W)$, consisting
of an arena $A=(V,V_\exists,E)$ together with two objectives
$\alpha_S$ and $\alpha_W$,
called the \emph{strong} and \emph{weak} objective, respectively; we also
refer to such games as $\alpha_S$ / $\alpha_W$ obliging games.
A strategy $\sigma$ for player $\exists$ is \emph{graciously winning}
for $v\in V$ if all plays that start at $v$ and
are compatible with $\sigma$ satisfy the strong objective $\varphi_S$ and furthermore
every finite prefix $\pi\in A^*$ of a play that is compatible with $\sigma$ can
be extended to an infinite play $\pi\tau$ that
is compatible with $\sigma$ and satisfies the weak objective $\alpha_W$.
We sometimes refer to the infinite plays $\pi\tau$ witnessing the
satisfaction of the weak objective as \emph{obliging}
plays. It follows immediately that player $\exists$ can only
win graciously at nodes at which at least one obliging play (satisfying $\alpha_S\wedge\alpha_W$)
starts, so we assume without loss of generality that this is the case for all nodes in $V$.\medskip

\begin{example}\label{ex:og} 
We consider the Emerson-Lei obliging game depicted
  below with the set $\{a,b,c,d\}$ of colors, a Streett condition
  (with two pairs $(a,b)$ and $(c,d)$) as strong objective $\alpha_S$,
  and generalized B\"uchi condition enforcing visiting both $a$ and $c$ as weak objective
  $\alpha_W$.
  In all examples in this work, nodes belonging to player
  $\exists$ are depicted with rounded corners, while $\forall$-nodes
  are depicted by rectangles; edges may have several colors in
  Emerson-Lei games, but in this example, each edge has at most one
  color. We use the dashed edge to illustrate both winning and losing
  in an obliging game. 
	\begin{center}
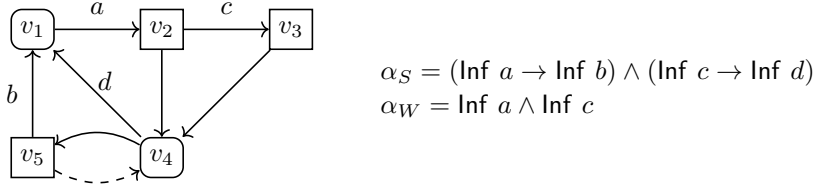
\begin{minipage}{5cm}
\tikzset{every state/.style={minimum size=15pt}, every node/.style={minimum size=15pt}}
    \begin{tikzpicture}[
		auto,
    node distance=1.7cm,
    semithick
    ]
     \node[draw, rounded corners] (1)  {$v_1$};
     \node[state, rectangle] [right of=1] (2) {$v_2$};
     \node[state, rectangle] [right of=2] (3) {$v_3$};
     \node[state, rectangle] [below of=1] (4) {$v_5$};
     \node[draw, rounded corners] [below of=2] (5)  {$v_4$};

     \path[->] (1) edge node {$a$} (2);
     \path[->] (2) edge node {} (5);

     \path[->] (2) edge node {$c$} (3);
     \path[->] (3) edge node {} (5);

     \path[->] (4) edge node {$b$} (1);
     \path[->] (5) edge node[right, pos=0.65] {$d$} (1);

     \path[->, dashed] (4) edge[bend right] node {} (5);
     \path[->] (5) edge[bend right] node {} (4);

  \end{tikzpicture}
\end{minipage}
\begin{minipage}{7cm}
$\alpha_S=(\mathsf{Inf}~a\to\mathsf{Inf}~b)\wedge
(\mathsf{Inf}~c\to\mathsf{Inf}~d)$\\
$\alpha_W=\mathsf{Inf}~a \wedge\mathsf{Inf}~c$
\end{minipage}
\end{center}
Consider the strategy $\sigma$ with which player $\exists$ alternatingly moves to $v_5$ and to $v_1$ when node $v_4$ is reached, depending on where they moved from $v_4$ the last time it has been visited. Without the dashed edge, $\sigma$ is graciously winning: every play compatible with $\sigma$ visits the colors
$a,b$ and $d$ infinitely often and hence satisfies $\alpha_S$; also, $\sigma$ allows player $\forall$ to visit color $c$ arbitrarily often by moving from $v_2$ to $v_3$, so every finite prefix of a $\sigma$-play can be continued to an obliging $\sigma$-play that infinitely often visits $v_3$ and
therefore satisfies $\alpha_S\wedge\alpha_W$. 
If the dashed edge is added to the arena, then $\sigma$ is no longer graciously winning.
Indeed, when playing against $\sigma$, player $\forall$ can prevent $b$ from ever being visited by always moving from $v_5$ to $v_4$.
However, the modified strategy $\sigma'$ that
moves from $v_4$ to $v_1$ only if the last visited node
is \emph{not} $v_5$ and also $b$ has been visited more recently than $d$ (and otherwise moves back to $v_5$) is graciously winning.
Every $\sigma'$-play that ends in $(v_4 v_5)^\omega$ satisfies
$\alpha_S$ but also can be made into an obliging play 
by having $\forall$ move to $v_1$ whenever $v_5$ is reached.
\end{example}

\noindent
\emph{Emerson-Lei automata.} Given an Emerson-Lei objective
$\alpha_{\gamma_C,\varphi_C}$, an \emph{Emerson-Lei automaton}
is a tuple $A=(\Sigma,Q,\delta,q_0,\alpha_{\gamma_C,\varphi_C})$,
where $\Sigma$ is the alphabet, $Q$ is a set of \emph{states},
$\delta\subseteq Q\times\Sigma\times Q$ is the transition relation,
and $q_0\in Q$ is the initial state;
in this context, we assume that $\gamma_C:\delta\to 2^C$ assigns
sets of colors to transitions in $A$.
A \emph{run} of $A$ on some infinite word $w=a_0 a_1\ldots\in \Sigma^\omega$ is a sequence $\pi=q_0 q_1\ldots \in Q^\omega$
such that $(q_i,a_i,q_{i+1})\in \delta$ for all $i\geq 0$.
A run $\pi$ is \emph{accepting} if and only if $\gamma_C(\pi)\models \varphi_C$, and $A$ \emph{recognizes} the language
\begin{align*}
L(A)=\{w\in \Sigma^\omega\mid \text{there is an accepting run of $A$ on $w$}\}.
\end{align*}
An Emerson-Lei automaton $A$ is \emph{non-empty} if and only if $L(A)\neq\emptyset$.
All automata that we consider in this work will have a single-letter alphabet $\Sigma=\{*\}$ so that they can read just the single infinite word $*^\omega$.

\section{Determinacy of Obliging Games}\label{sec:determinacy}

Chatterjee et al.~\cite{ChatterjeeHL10}, when formalizing obliging
games, stated that the standard shape of strategies (that is, $ A^* \cdot V_i\to V$) does not allow
player $\forall$ to counteract player $\exists$'s strategy with a
single strategy.
We show that, for a more general form of strategy (of type $A^\omega \to V\cup A^\omega$), this is possible.
Furthermore, with these more general strategies, the games are
determined.
That is, players are not disadvantaged by revealing their strategy
first and one of the players always has a winning strategy.
This insight allows offering alternative (more efficient) solutions
to obliging games in the next sections.

Fix an obliging game $G=(A,\alpha_S,\alpha_W)$.
Let $A=(V,V_\exists,E)$ and let $\alpha_S$ and $\alpha_W$ be Borel sets.
A \emph{witness} for vertex $v\in V$ is an infinite play $\pi=v_0v_1\cdots \in
A^\omega$ such that $v_0=v$; we write $\mathsf{witness}(v)$ to denote
the set of witnesses for $v\in V$.
The \emph{witness game} $W(G)=(A',\alpha')$ captures the obliging game
by allowing player $\exists$ to choose an explicit witness for a
given vertex. Player $\forall$ can then either check whether the
witness satisfies both $\alpha_S$ and $\alpha_W$, or stop at some point verifying the
witness and ask for a new witness.
If player $\forall$ changes witnesses infinitely often, they still check that
the resulting play satisfies $\alpha_S$.

Formally, we have
$A'=(V',V'_\exists,E')$, where
$V'=V \cup  A^\omega$, $V'_\exists = V$ and $E'$ is as follows.
	$$
	\begin{array}{l c l l}
	E' & = &
	\{(v,\pi)\mid v\in V, \pi\in\mathsf{witness}(v)\} \qquad \cup \qquad \{(v\pi,\pi)\mid v\pi\in A^\omega\} & \cup
	\\
   & & 	\{(v\pi,v'') ~|~ v\pi\in A^\omega, v \in V_\forall, 
	v''\in E(v)
\}.
\end{array}$$

	Given $v\pi \in V'_\forall$ put $\lft(v\pi)=v$.
	We extend $\lft$ to sequences over $V'$ in the natural way.
	Consider a play $\pi'\in (A')^\omega$.
	Let $\pi'\Downarrow
	_{V'_\forall}$ denote the projection of $\pi'$ to the elements of
	$V'_\forall$,
	that is, $\pi'\Downarrow
	_{V'_\forall}$ is obtained from $\pi'$ by removing all elements in $V'_{\exists}=V$.
	Then put $\seq_A(\pi')=\lft(\pi'\Downarrow_{V'_\forall})$.
	Clearly, $\seq_A(\pi') \in A^\omega$,
	that is, $\seq_A(\pi')$ extracts from $\pi'$ the infinite
	play in $A$ that is followed in $\pi'$.
	The winning condition $\alpha'$ consists of plays $\pi'$ that remain
	eventually in $V'_{\forall}$ forever such that $\seq(\pi')$
	satisfies both $\alpha_S$ and $\alpha_W$, or plays $\pi'$ that visit
	$V'_\exists$ infinitely often such that $\seq(\pi')$
	satisfies $\alpha_S$.
	Formally, we have the following:
	$$
	\begin{array}{l c l l}
	\alpha' & = & \{ \pi'\in V'^*\cdot (V'_{\forall})^\omega ~|~
	\seq_A(\pi') \in \alpha_S\cap \alpha_W\} & \cup \\
	& & \{ \pi'\in V'^*\cdot (V'_\exists \cdot (V'_{\forall})^*)^\omega
	~|~
	\seq_A(\pi') \in \alpha_S \}
	\end{array}
 	$$

\begin{example}
For brevity, we refrain from showing the complete witness game associated to the obliging game from Example~\ref{ex:og} and instead consider just
two witnesses for the node $v_1$, namely 
 $(v_1 v_2 v_4 v_5)^\omega$ and
 $(v_1 v_2 v_3 v_4 v_1 v_2 v_4 v_5)^\omega$.
The former satisfies the Streett objective $\alpha_S$ as it visits the colors $a$ and $b$.
However, it does not satisfy the generalized B\"uchi objective $\alpha_W$ as color $c$ is not visited infinitely often.
The latter witness, contrarily, visits all colors infinitely often and hence satisfies $\alpha_S\wedge\alpha_W$.
In the witness game, player $\exists$ can move from $v_1$ to these two
witnesses (and to many more). Player $\forall$ in turn can win the
first witness by exploring it indefinitely, thereby showing that it
does not satisfy $\alpha_W$; doing the same for the second witness,
player $\forall$ loses. In both certificates, we have exit moves
from every position such that the node at the position is owned by
player $\forall$, that is, moves from positions with node $v_2$ to
$E(v_2)$, and similarly for $v_3$ and $v_5$.   
\end{example}

\begin{lemma}
Player $\exists$ is graciously winning in $G$ at $v$ iff
player $\exists$ is winning in $W(G)$ at $v$.
\label{lemma:gracious games}
\end{lemma}

\begin{corollary}
Obliging games are determined.
\label{corr:determinacy of obliging games}
\end{corollary}

\begin{proof}
	This follows immediately from Lemma~\ref{lemma:gracious games} and
	the determinacy of games with Borel winning conditions
	\cite{Martin75}.
	Notice that Martin's determinacy result holds also for games with
	continuous vertex-spaces and continuous branching degrees, as in
	the witness game.
\end{proof}

\section{Reducing $\omega$-Regular Obliging Games to Finite Games}
\label{sec:certgames}

From this point on, we restrict our attention to obliging games in which both objectives are Emerson-Lei objectives; we note that
every $\omega$-regular objective can be transformed to an Emerson-Lei objective.
It turns out that due to the $\omega$-regularity of Emerson-Lei objectives, obliging plays in such games have finite witnesses
that take on the form of lassos that are built over the game arena at
hand. Formally, we show that for every such game we can create a game that is smaller than the witness game by restricting
attention to witnesses of this specific form that satisfy the
acceptance conditions. We call such witnesses
\emph{certificates}, which we define next.

\subsection{Certificates from Witnesses}
We fix an Emerson-Lei obliging game $G=(A,\alpha_S,\alpha_W)$
with arena $A=(V,V_\exists,E)$, objectives $\alpha_S=(\gamma_S,\varphi_S)$ and
$\alpha_W=(\gamma_W,\varphi_W)$, and put $n:=|V|$, $d:=|S|$ and $k:=|W|$.

\begin{definition}[Certificate]
Given a node $v\in V$, a \emph{certificate} for $v$ (in $G$) is a witness for $v$ that is
of the form $c=wu^\omega$; if $c$ satisfies $\varphi_S\land\varphi_W$, then we
say that $c$ is a \emph{valid certificate}.
\end{definition}

We equally represent $c=wu^\omega$ by the pair $c=(w,u)$.
Let $w=w_0 w_1\ldots
w_m$ and $u=u_0 u_1\ldots u_r$.
We refer to $w$ as the \emph{stem} and to $u$ as the \emph{loop} of
$c$, and to $m$ and $r$
as the length of the stem and the loop, respectively.
When the partition of $c$ is not important we sometimes just write
$c=v_0\ldots v_{m+r+1}$.
Clearly, as satisfaction of Emerson-Lei objectives depends only on the
infinite suffixes of a play, it follows that in a valid certificate, $u^\omega$ also satisfies
$\varphi_S$ and $\varphi_W$.

Given a coloring function $\gamma_C$ over some set $C$ of colors,
the $C$-\emph{fingerprint} of a finite play $\pi=v_0 v_1\ldots v_j\in A^*$ is
the set $\textstyle\bigcup_{0\leq i< j}\gamma_C(v_i,v_{i+1})$ of colors
visited by $\pi$,
according to $\gamma_C$.\medskip

Next we show that given a witness in $G$ that satisfies $\alpha_S\land\alpha_W$, we can alway construct a valid certificate of
size at most
$\mathsf{certLen}:=n\cdot d + (d+k+1)\cdot (n+1)\in \mathcal{O}(n\cdot(\max(d,k)))$.

\begin{lemma}[Certificate existence]\label{lem:certificateExistence}
Let
$v\in V$ and let $\pi=\pi_0 \pi_1\ldots$ be a witness for $\pi_0=v$
that satisfies $\varphi_S\land\varphi_W$.
Then there is a valid certificate $c=(w,u)$ for $v$ with
stem length at most $n\cdot d$ and loop length at most $(d+k+1)\cdot (n+1)$, such that
for all positions $i$ in $c$ there is a position $j$ in $\pi$ such
that $v_i=\pi_j$ and the $S$-fingerprints of $v_0\ldots v_i$ and $\pi_0\ldots \pi_j$ coincide.
\end{lemma}

We let $\mathsf{Cert}(v)$ and $\mathsf{Cert}$ denote the sets
of all valid certificates for some $v\in V$ in $G$ and all valid
certificates for all $v\in V$ in
$G$, respectively,
subject to the size bounds obtained in Lemma~\ref{lem:certificateExistence}.
Then we have
$|\mathsf{Cert}(v)|\leq|\mathsf{Cert}|\leq n^{\mathsf{certLen}}\in 2^{\mathcal{O}(n\cdot (\max(d,k))\cdot \log n)}$.

\subsection{Certificate Games and Smaller Winning Strategies}

Next we adapt the witness games from Section~\ref{sec:determinacy}
to use finite certificates in place of witnesses, which leads to \emph{certificate games} that have finite game
arenas. In a certificate game, player $\exists$ has to pick
a single valid certificate at each node $v\in V$, thereby commiting to a long-term future behavior that satisfies $\alpha_S\land\alpha_W$. Player
$\forall$ in turn can either accept the certificate
(and thereby lose the play), or pick some position $i$ in the certificate and challenge whether player $\exists$ still
has a strategy to graciously win when player $\forall$ exits the certificate at position $i$. All plays then either end with player
$\forall$ eventually losing by accepting some certificate, or player $\forall$ infinitely often exits certificates, in which case the winner of the play
is determined using just the strong objective $\alpha_S$.\medskip

Formally, the certificate game associated to $G$ is
a (non-obliging) Emerson-Lei game $C(G)=(B,\beta)$
with arena $B=(N,N_\exists,R)$. The game is played over the sets
$N_\exists=V$ and $N_\forall = \mathsf{Cert}\cup \mathsf{Cert}\times [\mathsf{certLen}]\times V$ of nodes.
The moves in $C(G)$ are defined by
\begin{align*}
R(v)&=\mathsf{Cert}(v) &
R(c)&=\{(c,i,v)\in N_\forall\mid
v_i\in V_\forall, v\in E(v_i)\}\cup\{c\} &
R(c,i,v)&= \{v\}
\end{align*}
for $v\in V$, $c=v_0 v_1\ldots v_m\in\mathsf{Cert}$ and $i\in[\mathsf{certLen}]$.
Thus player $\exists$ has to provide a valid certificate for $v$ whenever a play reaches a
node $v\in V$. Player $\forall$ in turn can challenge the way
certificates are combined by exiting them from universal nodes in
their stem or loop,
and continuing the game at the exit node; if a certificate $c$ does not
contain a universal node, then player $\forall$ has to take the loop
at $c$, intuitively accepting the certificate $c$.
Intermediate nodes $(c,i,v)$ are used to make explicit the $S$-fingerprint
of the path through a certificate $c$ that is taken before exiting it by moving from $v_i$ to $v$; this is necessary since a certificate may contain universal nodes that allow moving from different positions in the certificate to a single exit node $v$, potentially giving player $\forall$ a choice on the path that is taken (and on the $S$-fingerprint that is accumulated) through the certificate before exiting to $v$.

As our notion of games does not support deadlocks,
we annotate trivial loops with an additional color $c_\top$ to
encode the situation that player $\forall$ accepts a certificate
and thereby loses.
The coloring function $\gamma':R\to 2^{S\cup \{c_\top\}}$
also keeps track of the $S$-fingerprints when passing through certificates and is defined by
$\gamma'(c,(c,i,v)) = \gamma_S(v_0\ldots v_i v)$,
$\gamma'(v,c) = \emptyset$ and $\gamma'((c,i,v),v)=\emptyset$
for the non-looping moves and
by $\gamma'(c,c)=c_\top$ for looping moves at certificates $c\in\mathsf{Cert}$.
We then define
$\varphi'=\mathsf{Inf}(c_\top)\vee\varphi_S$,
intuitively expressing that player $\exists$ can win either according to $\varphi_S$ or by forcing player $\forall$ to get stuck in a trivial loop. Then we define $\beta$ to be the objective induced by $\varphi'$ and $\gamma'$.

\begin{example}\label{ex:certgames}

Below we depict the construction for a single node $v\in V$,
showing, as examples, just two valid certificates $c_1=v_0 v_1 v_2\in\mathsf{Cert}(v)$ and
$c_2=w_0 w_1 w_2 w_3\in\mathsf{Cert}(v)$ (in particular, we assume $v_0=w_0=v$ and that the certificates $c_1$ and $c_2$ satisfy both $\varphi_S$ and $\varphi_W$).
We further assume that $v_2\in V_\forall$, $E(v_2)=\{v_1,x\}$, and that $w_1=w_3\in V_\forall$ and $E(w_1)=E(w_3)=\{w_2\}$;
all other nodes in $c_1$ and $c_2$ are assumed to be contained in $V_\exists$.
\vspace{-5pt}
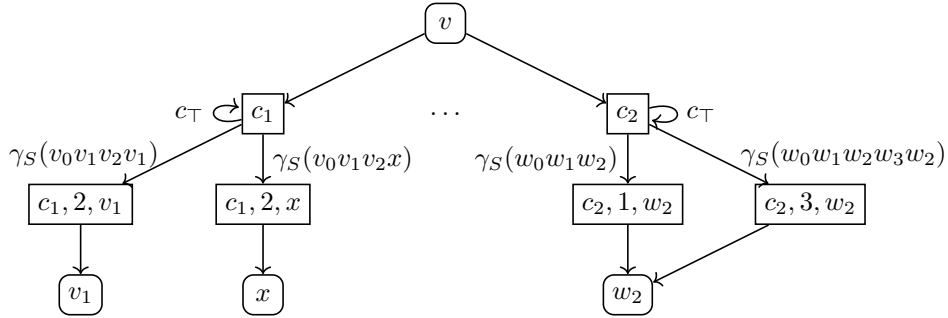
	\begin{center}
\tikzset{every state/.style={minimum size=15pt}, every node/.style={minimum size=15pt}}
    \begin{tikzpicture}[
		auto,
    node distance=1.2cm,
    semithick
    ]
     \node[draw, rounded corners] (2)  {$v$};
     \node[] (yo) [below of=2] {$\ldots$};
     \node[] (yoa) [left of=yo] {};
     \node[] (yob) [right of=yo] {};
     \node[state, rectangle] [left of=yoa] (3) {$c_1$};
     \node[state, rectangle] [right of=yob] (4) {$c_2$};

     \node[state,rectangle] [below of=3] (6) {$c_1,2,x$};
     \node[] (yo1) [left of=6] {};
     \node[state,rectangle] [left of=yo1] (5) {$c_1,2,v_1$};

     \node[draw, rounded corners] [below of=5] (10) {$v_1$};
     \node[] (yo2) [right of=10] {};
     \node[draw, rounded corners] [right of=yo2] (9) {$x$};

     \node[draw, rectangle] [below of=4] (7) {$c_2,1,w_2$};
     \node[] (yo3) [right of=7] {};
     \node[draw, rectangle] [right of=yo3] (8) {$c_2,3,w_2$};

     \node[draw, rounded corners] [below of=7] (11) {$w_2$};

     \path[->] (2) edge node {} (3);
     \path[->] (2) edge node {} (4);


     \path[->] (3) edge [loop left] node {$c_\top$} (3);
     \path[->] (4) edge [loop right] node {$c_\top$} (4);

     \path[->] (6) edge node {} (9);
     \path[->] (5) edge node[left] {} (10);

     \path[->] (7) edge node[left] {} (11);
     \path[->] (8) edge node {} (11);

     \path[->] (3) edge node[left,pos=0.5] {$\gamma_S(v_0 v_1 v_2 v_1)\,\,\,$} (5);
     \path[->] (3) edge node[right,pos=0.5] {$\gamma_S(v_0 v_1 v_2 x)$} (6);
     \path[->] (4) edge node[left,pos=0.5] {$\gamma_S(w_0 w_1 w_2)$} (7);
     \path[->] (4) edge node[right,pos=0.5] {$\,\,\,\,\,\gamma_S(w_0 w_1 w_2 w_3 w_2)$} (8);

  \end{tikzpicture}
\end{center}
\vspace{-5pt}

At node $v$, player $\exists$ has to provide some valid certificate
for $v$; assuming that
player $\exists$ picks the certificate $c_1=v_0 v_1 v_2$,
player $\forall$ in turn can exit the certificate at position $2$ (as
$v_2\in V_\forall$) by moving to either $(c_1,2,v_1)$ or $(c_1,2,x)$ (as $E(v_2)=\{v_1,x\}$)
thereby triggering $S$-fingerprint $\gamma_S(v_0 v_1 v_2 v_1)$ or $\gamma_S(v_0 v_1 v_2 x)$; player $\forall$ intuitively
cooperates during the play that leads from $v_0$ to $v_{2}$, but may stop cooperating at node $v_2$ by moving to either $v_1$ or $x$. Similarly, player $\forall$ can exit the certificate $c_2$ at positions $1$ or $3$, both leading to the node $w_2$, but with different $S$-fingerprints, that is, with different sets of visited colors.

\end{example}

As the above example shows, certificate games have a three-stepped structure. Plays progress from nodes $v$ of the original game on to certificate nodes $c$, and then onwards to nodes $(c,i,v')$ that encode the part of $c$ that is visited before exiting $c$, and then proceed back to some node $v'$ of the original game. We refer to these three-step subgames a \emph{gadgets}; each gadget has a  starting node $v$ and a (possibly empty) set of exiting nodes. We point out that, crucially, gadgets do not contain (non-trivial) loops.

We state the correctness of the reduction to certificate games as follows.

\begin{theorem}\label{theorem:gameCorrectness}
Let $G$ be an Emerson-Lei obliging game and let $v$ be a node in $G$, and
recall that $W(G)$ and $C(G)$ respectively
are the witness game and the certificate game associated with $G$.
The following are equivalent:
\begin{enumerate}
\item player $\exists$ graciously wins $v$ in $G$;
\item player $\exists$ wins $v$ in $W(G)$;
\item player $\exists$ wins $v$ in $C(G)$.
\end{enumerate}
\end{theorem}
\begin{proof}
The equivalence of the first two items is stated by Lemma~\ref{lemma:gracious games} and the proof of the implication from item two to item three is technical but straight-forward.
Therefore we show just the implication from item three to item one,
which also shows how to construct a graciously winning strategy in $G$
from a winning strategy in $C(G)$. 
For this proof, we use
\emph{strategies with memory}, introduced next. 
A strategy for player
$i\in\{\exists,\forall\}$ \emph{with memory $M$} is a tuple
$\sigma=(M,m_0,\mathsf{update}:M\times E\to M,\mathsf{move}:V_i\times M\to V)$,
where $M$ is some set of \emph{memory values}, $m_0\in M$ is the initial memory value,
the \emph{update function} $\mathsf{update}$ assigns the outcome $\mathsf{update}(m,e)\in M$ of updating the
memory value $m\in M$ according to the effects of taking the move $e\in E$, and
the moving function $\mathsf{move}$ prescribes a single move $(v,\mathsf{move}(v,m))\in E$ to
every game node $v\in V_i$ that is owned by player $i$, depending on the memory value $m$.
We extend $\mathsf{update}$ to finite plays $\pi$ by putting
$\mathsf{update}(m,\pi)=m$ in the base case that $\pi$ consists of a
single node,
and by putting $\mathsf{update}(m,\pi)=\mathsf{update}(\mathsf{update}(m,\tau),(v,w))$
if $\pi$ is of the shape $\tau v w$, that is, contains at least two nodes.

Let $\sigma=(M,m_0,\mathsf{update},\mathsf{move})$ be a
winning strategy with memory $M$ for player $\exists$ in the game $C(G)$.
We construct a strategy
$\tau=(M',(m_0,v,1),\mathsf{update}',\mathsf{move}')$
with memory $M'=V\times M\times \{1,\ldots,2|S|+|W|\}$
for player $\exists$ in the game $G$ such that $\tau$ graciously wins $v$.
To this end, we note that $\sigma$ provides, for each node $w\in V$ and memory value $m\in M$,
a certificate $c(w,m):=\mathsf{move}(w_\exists,m)\in \mathsf{Cert}(w)$ for $w$ in $G$.

The strategy $\tau$ uses memory values $(w,m,i)$, where
$w$ and $m$ identify a current certificate $c(w,m)$ and $i$ is a counter used for the construction of
this certificate.
We define $\tau$ such that player $\exists$ starts by building the
certificate $c(v,m_0)$ for $v$ and $m_0$. This process continues
as
long as player $\forall$ obliges, that is, does not move outside
of this certificate. Assuming that player $\forall$ obliges, the memory required to construct
the certificate is bounded by $2|S|+W$, walking, in the prescribed order, through the certificate.
In the case that player $\forall$ eventually stops obliging and takes a move to some node $w$ that
is not the next node on the path prescribed by the certificate, the memory
for the strong objective is updated according to the play from $v$ to $w$, resulting in a new memory
value $m$, and the memory for certificate construction is reset. Then the certificate
construction starts again, this time for the certificate $c(w,m)$ prescribed by $\sigma$ for the new starting values $w$ and $m$.

In more detail, every node from $V$ is visited at most $2|S|+|W|$ times within a certificate $c(v,m)=(v_0 v_1\ldots v_n)$ before the end of the loop $v_n$ is reached. To each position $i$
within $c(v,m)$, we associate the number $\sharp i$ such that
$v_i$ is the $\sharp i$-th occurence of the node $v_i$ in $c(v,m)$; we note that for all $1\leq i\leq n$, we have
$1\leq \sharp i\leq 2|S|+|W|$. Conversely, given a node $w$ and a number $j$ between $1$ and $2|S|+|W|$, we let $\mathsf{pos}(w,j)$ denote the position of the $j$-th occurence of $w$ in 
$c(v,m)$. In the case $w$ occurs less than $j$ times in $c(v,m)$, we leave $\mathsf{pos}(w,j)$ undefined; below we make sure that this value is always defined when it is used.

Let $j$ be the starting position of the loop of $c(v,m)$.
Given a position $i$ in $c(v,m)$, we abuse notation to let 
$i+1$ denote just $i+1$ if $i<n$, and to denote $j$ if $i=n$; in this way, we encode taking single steps within a certificate, wrapping back to the start of the certificate loop, once the end of the loop is reached.
  
We define the strategy $\tau$ to always move to the next node in the current certificate, using the memory value $i$ together with the current node $w$ to find the current position in the certificate $c(v,m)$, that is,
we put
\begin{align*}
\mathsf{move}'(w,(v,m,i))=
v_{\mathsf{pos}(w,i)+1} 
\end{align*}
for $w\in V_\exists$ and $(v,m,i)\in M'$.
The memory update in $\tau$ incorporates the memory update function from $\sigma$, but additionally also keeps track of the memory for certificate construction. For moves $(w,w')$ that stay within the current certificate (that is, $w'= v_{\mathsf{pos}(w,i)+1}$), this memory is updated according to proceeding one step within the certificate; for moves that leave the current certificate, the memory for certificate construction is reset while the memory for $\sigma$ is updated according to the path taken through the certificate before exiting it. Formally, we put
\begin{align*}
\hspace{-10pt}
\mathsf{update}'((v,m,i),(w,w'))=\begin{cases}
(v,m,\sharp (\mathsf{pos}(w,i)+1)) & \text{$w'= v_{\mathsf{pos}(w,i)+1}$}\\
(w',\mathsf{update}(m,(v_0v_1\ldots ww')),1) & \text{$w'\neq v_{\mathsf{pos}(w,i)+1}$}
\end{cases}
\end{align*}
for $(v,m,i)\in M'$ such that $c(v,m)=v_0 v_1 \ldots v_n$
and $(w,w')\in E$; notice that in plays that adhere to $\tau$,
 the latter case can, by definition of $\mathsf{move}'$, only happen for $w\in V_\forall$.

To see that player $\exists$ graciously wins $v$ using the constructed strategy $\tau$, let $\pi$ be a play that adheres to $\tau$. Then
$\pi$ either eventually stays within one certificate $c(w,m)$ forever, or $\pi$
induces an infinite play $\rho$ of $C(G)$ that adheres to $\sigma$.
In the latter case, $\pi$ changes the certificate infinitely often, and the $S$-fingerprints of $\pi$ and $\rho$ coincide by construction, showing that $\pi$ satisfies $\varphi_S$; if $\pi$ eventually
stays within one certificate $c(w,m)$ forever, we note that the loop of $c(w,m)$
satisfies $\varphi_S$ so that $\pi$ satisfies $\varphi_S$ as well. Thus every play that adheres
to $\tau$ satisfies the strong objective $\varphi_S$.
 Next, let $\pi_f$ be a finite prefix of some
play that adheres to $\tau$, and let $\pi_f$ end in some node $w$. We have to show that there is
an infinite play $\pi$ that adheres to $\tau$, extends $\pi_f$ and satisfies $\varphi_W$.
Let $m$ be the value of the memory for the strong objective at the end of $\pi_f$. We consider
the play of $G$ in which player $\forall$ obliges from the end of $\pi_f$ on, forever.
By construction, this play is $\pi_f$ extended with the certificate $c(w,m)$
which adheres to $\tau$ and satisfies $\varphi_W$.
\end{proof}

The strategy construction given in the proof of Theorem~\ref{theorem:gameCorrectness} yields:

\begin{corollary}\label{cor:stratSize}
Given an obliging game with Emerson-Lei objectives $\alpha_S$ and $\alpha_W$ and $n$ nodes that contains a node $v$ at which $\exists$ is graciously winning, there is a graciously winning strategy for $v$
that uses memory at most $n\cdot (2|S|+|W|)\cdot m$, where $m$ is the amount of memory required by winning
strategies for player $\exists$ in standard games with objective $\alpha_S$.
\end{corollary}

\begin{remark}[Canonical certificates]\label{rem:canon}
With the proposed strategy extraction, certificate strategies memorize the starting point of the certificate that player $\exists$ currently attempts to construct, leading to an additional linear factor $n$ in strategy size. We conjecture that
winning strategies in certificate games can be transformed to make them \emph{reuse} certificates (so that the choice of the certificate does not depend on the starting point). Strategies
with such canonical certificate choices would allow for removing the additional factor $n$ in Corollary~\ref{cor:stratSize}.
\end{remark}

In particular, our result shows the existence of graciously winning strategies with quadratic sized strategies for
all obliging games that have a half-positional strong objective (i.e. one for which standard games have positional winning strategies for player $\exists$); this covers, e.g., obliging games in which $\alpha_S$ is a parity or Rabin objective.

In the table below we compare the solution via certificate games to a previous
solution method~\cite{ChatterjeeHL10,MajumdarSchmuck23} reduces $\alpha_S$ / $\alpha_W$ obliging games to standard games with an objective of type $\alpha_S \land $ B\"uchi.
In this approach,
the weak objective is encoded by a non-deterministic B\"uchi
automaton and graciously winning strategies obtained in this way incorporate the state space of the automaton. The encoding of $\alpha_W$ by a B\"uchi automaton leads to
linear blow-up if $\alpha_W$ is a Rabin objective, but is exponential if $\alpha_W$ is a Streett or general Emerson-Lei objective.

In contrast to this, the certificate games that we propose here always have (essentially) just $\alpha_S$ as objective, and the dependence
of strategy size on the weak objective $\alpha_W$ is in all cases linear in $|W|$. In comparison to~\cite{ChatterjeeHL10}, our approach hence leads to significantly smaller strategies for obliging games in which $\alpha_W$ is at least a Streett objective. In the cases
where $\alpha_W$ is a Rabin objective, strategies obtained from certificate games are slightly larger than in~\cite{ChatterjeeHL10};  we conjecture that this can be improved by sharing
certificates (cf. Remark~\ref{rem:canon}).

\begin{table}
\begin{center}
\begin{footnotesize}
\begin{tabular}{| l | c | c | l| r |}
\hline
type of $\alpha_S$ & type of $\alpha_W$ & objective red.~\cite{ChatterjeeHL10} & strategy size~\cite{ChatterjeeHL10} & strategy size\\
\hline
\hline
\multirow{3}{*}{parity($d$)} 
 & Rabin($k$) & \multirow{3}{*}{parity($d$)$\land$ B\"uchi} & $d(4k+2)$ &  $(2d+2k)n$ \\\cline{2-2}\cline{4-5}
 & Streett($k$) &  & $2^{k+1}dk$ &  $(2d+2k)n$ \\\cline{2-2}\cline{4-5}
 & EL($k$) & & $2^{k+1}dk$ &  $(2d+k)n$ \\
\hline
\hline
\multirow{3}{*}{Rabin($d$)} & Rabin($k$) & \multirow{3}{*}{Rabin($d$)$\land$ B\"uchi} & $d(4k+2)$ &  $(4d+2k)n$ \\\cline{2-2}\cline{4-5}
& Streett($k$) & & $2^{k+1}dk$ &  $(4d+2k)n$  \\\cline{2-2}\cline{4-5}
& EL($k$) & & $2^{k+1}dk$ &  $(4d+k)n$  \\
\hline
\hline
\multirow{3}{*}{Streett($d$)} & Rabin($k$) & \multirow{3}{*}{Streett($d+1$)}&$(d+1)!(4k+2)$ & $(4d+2k)d!n$ \\\cline{2-2}\cline{4-5}
& Streett($k$) & & $(d+1)!2^{k+2}k $ & $(4d+2k)d!n$  \\\cline{2-2}\cline{4-5}
& EL($k$) & & $(d+1)!2^{k+2}k$ & $(4d+k)d!n$  \\
\hline
\hline
EL($d$) & EL($k$) & EL($d+1$) & $(d+1)!2^{k+2}k$ &  $(2d+k)d!n$ \\
\hline

\end{tabular}
\end{footnotesize}
\end{center}
\caption{Comparison of upper bounds on strategy sizes for various types of obliging games.}
\end{table}

For instance, for obliging games with $\alpha_S=$ Rabin($d$) and $\alpha_W=$ Streett($k$), the approach from~\cite{ChatterjeeHL10} uses nondeterministic B\"uchi automata with $2^kk$ states to encode the weak
(Streett) objective. The reduction leads to a game with
objective Rabin$(d)\wedge$ B\"uchi; winning strategies in such games require $2d$ additional memory values for each automaton state (cf.~\cite{DziembowskiJW97}), resulting in an overall memory requirement of $2^{k+1}dk$.
In contrast, the reduced certificate
game in this case is a Rabin$(d)$ game, and the extracted
gracious strategies require memory only to identify the current certificate and a position in it; the overall memory requirement
for strategies obtained through our approach thus is $(4d+2k)n$.

\section{Solving Certificate Games Efficiently}\label{sec:solving}

In the previous section we have shown how obliging games
with Emerson-Lei objectives $\alpha_S$ and $\alpha_W$ can be
reduced to certificate games.
This reduction not only yields games with a simpler objective (essentially just $\alpha_S$) than in the previously known alternative reduction from~\cite{ChatterjeeHL10},
but it also
shows that the memory required for graciously winning strategies
in obliging games always depends only linearly on $|W|$. However,
the proposed reduction to certificate games
makes explicit all possible certificates that exist in the original obliging game and therefore
incurs exponential blowup.
In more detail, given an Emerson-Lei obliging game $G$ with $n$ nodes and sets $S$ and $W$ of colors,
the certificate game $C(G)$ from the previous section
is of size $2^{\mathcal{O}(n\cdot\max(|S|,|W|)\cdot \log n)}$ and uses $|S|$ many colors; solving it naively does not improve upon previously known solution algorithms for obliging games.

In this section, we show that the gadget constructions in certificate games essentially encode non-emptiness checking for non-deterministic $\omega$-automata; intuitively, a certificate for a game node
is a witness for the non-emptiness of an Emerson-Lei automaton
with acceptance condition $\alpha_S\wedge\alpha_W$ that lives
over (parts of) the original game arena.
We show that it suffices to check for the existence of a single suitable certificate, rather than exploring all possible certificates that occur as explicit nodes in the reduced game.

As pointed out above, the gadget parts in a certificate game do not have non-trivial cycles (that is, we can regard them as directed acyclic graphs, DAGs). Consequently, it is possible to simplify the solution process of these (potentially) exponential-sized parts of the game by instead using non-emptiness checking for suitable $\omega$-automata. If for instance both $\alpha_S$ and $\alpha_W$ are Streett conditions (so that $\alpha_S\wedge\alpha_W$ again is a Streett condition), then the solution of every gadget part in the game can be reduced to non-emptiness checking of Streett automata. As non-emptiness checking for Streett automata can be done in polynomial time, this trick in effect removes the exponential runtime factor that originates from the large number of certificates when solving such games  (and in general, whenever non-emptiness of $\alpha_S\wedge\alpha_W$-automata can be checked efficiently).

Our program is as follows:
We first show how Emerson-Lei
certificate games can be reduced to
parity games using a tailored later-appearence-record (LAR)
construction, incurring blow-up $|S|!$ in game size, but
retaining the DAG structure of gadget subgames.
Importantly, this is required to retain the dependence of
non-DAG nodes on a small number of (post DAG) successors.
Then we show the relation between attractor computation in gadget subgames within the obtained parity games and the non-emptiness problem for specific Emerson-Lei automata.
Finally, we apply a fixpoint acceleration method from~\cite{Hausmann24}
to show that during the solution of parity games, the solution of DAG substructures can be replaced with a procedure that decides attraction to (subsets of) the exit nodes of the DAG. Overall, we therefore show that the winning regions of certificate games can be computed as nested fixpoints of a function that checks certificate existence by nonemptiness checking suitable Emerson-Lei automata.

\subsection{Lazy parity transform.}\label{subsec:lar}
We intend to transform $C(G)$ to an equivalent parity game,
using a lazy variant of the later-appearance-record (LAR)
construction (cf.~\cite{GurevichH82,HunterD05}).
To this end, we fix a set $C$ of colors and introduce notation
for permutations over $C$.
We let $\Pi(C)$ denote the set of permutations over $C$,
and for a permutation $\pi\in\Pi(C)$ and a position $1\leq i\leq |C|$, we let
$\pi(i)\in C$ denote the element at the $i$-th position of $\pi$.
For $D\subseteq C$ and $\pi\in\Pi(C)$, we let $\pi@D$ denote the permutation that is
obtained from $\pi$ by moving \emph{the} element of $D$ that occurs at the right-most position in
$\pi$ to the front of $\pi$; for instance, for $C=\{a,b,c,d\}$ and $\pi=(a,d,c,b)\in\Pi(C)$,
we have $\pi@\{a,d\}=\pi@\{d\}=(d,a,c,b)$ and $(d,a,c,b)@\{a,d\}=\pi$.
Crucially, restricting the reordering to single colors, rather than sets of colors, ensures that for
each $\pi\in \Pi(C)$ and all $D\subseteq C$, there are only $|C|$ many $\pi'$ such that $\pi@D=\pi'$.
 Given a permutation $\pi\in\Pi(C)$ and an index $1\leq i\leq |C|$, we furthermore let  $\pi[i]$ denote the set of colors that occur
in one of the first $i$ positions in $\pi$.

Next we show how permutations over $C$ can be used to transform
Emerson-Lei games with set $C$ of colors to parity games;
the reduction annotates nodes from the original game
with permutations that serve as a memory, encoding the
order in which colors have recently been visited.
The transformation is lazy as it just moves the most
significant color that is visited by a set of colors $C$
rather than the entire set.

\begin{definition}\label{defn:LAR}
Let $G=(A,\alpha_C)$ be an Emerson-Lei game with arena
$A=(V,V_\exists,E)$, set $C$ of colors and objective
$\alpha_C$ induced by $\gamma_C$ and $\varphi_C$.
We define the parity game
\begin{align*}
P(G)=(V\times\Pi(C),V_\exists\times\Pi(C),E',\Omega:E'\to\{1,\ldots,2|C|+1\})
\end{align*}
by putting
$E'(v,\pi)=\{(w,\pi@\gamma_C(v,w))\mid (v,w)\in E\}$
for $(v,\pi)\in V\times\Pi(C)$, as well as
$\Omega((v,\pi),(w,\pi'))=2p$ if $\pi[p]\models\varphi_C$ and
$\Omega((v,\pi),(w,\pi'))=2p+1$ if $\pi[p]\not\models\varphi_C$,
for $((v,\pi),(w,\pi'))\in E'$. In the definition of $\Omega((v,\pi),(w,\pi'))$, we write $p$ to denote the right-most position in $\pi$
that contains some color from $\gamma_C(v,w)$.
For a finite play $\tau=(v_0,\pi_0)(v_1,\pi_1)\ldots (v_n,\pi_n)$
of $P(G)$,
we let $p(\tau)$ denote the
right-most position in $\pi_0$ such that $\pi_0(p(\tau))\in\gamma_C(v_ 0v_1\ldots v_n)$.
The game $P(G)$ has $|V|\cdot |C|!$ nodes and $2|C|+1$ priorities.

\end{definition}

\begin{lemma}\label{lem:LAR}
For all $v\in V$ and $\pi \in \Pi(C)$, player $\exists$ wins
from $v$ in $G$ iff player $\exists$ wins from $(v,\pi)$ in $P(G)$.
\end{lemma}

Now we consider the parity game $P(C(G))$ that is obtained by applying
the above LAR construction to the certificate game $C(G)$ from Section~\ref{sec:certgames}.
It is a parity game with $2^{\mathcal{O}(n\cdot\max(|S|,|W|)\cdot \log n)}\cdot |S|!$ many nodes and
$2|S|+1$ priorities.
We recall that all nodes in $C(G)$ that are of the
shape $c$ or $(c,i,v)$ such that
$c,(c,i,v)\in N_\forall$ are inner nodes of gadget subgames that consist
of three layers. Each such subgame has exactly one entry node and at most $n$
exit nodes, all contained in $N_\exists=V$.
The LAR construction preserves this general structure as it simply annotates game
nodes with permutations of colors.
Specifically, $P(G)$ has $|S|!\cdot n$ entry and exit nodes
for all such subgames together.
Furthermore, for all entry nodes $(v,\pi)$ of a subgame in
$P(G)$, we have that
every exit node (that can be reached from $(v,\pi)$ with excactly three moves,
not accounting for trivial loops)
is of the shape $(w,\pi@i)$ where $w\in V$ and
$0\leq i\leq |C|$. While an entry node for an individual subgame in $C(G)$ has
at most $n$ exit nodes, each entry node for a subgame in $P(C(G))$
has at most $n(|S|+1)$ exit nodes.
Using the classical LAR would result in a potentially
exponential number of exit nodes.
Indeed, for every possible subset $C'\subseteq C$ there could
be a different exit node.

\begin{example}
Below, we depict (part of) the parity game that is obtained from the certificate
game fom Example~\ref{ex:certgames} by using the proposed LAR construction.
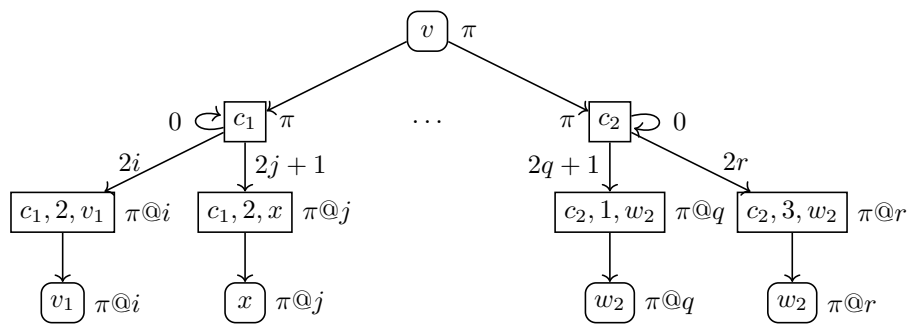
\begin{center}
\tikzset{every state/.style={minimum size=15pt}, every node/.style={minimum size=15pt}}
    \begin{tikzpicture}[
		auto,
    node distance=1.2cm,
    semithick
    ]
     \node[draw, rounded corners,label=right:{$\pi$}] (2)  {$v$};
     \node[] (yo) [below of=2] {$\ldots$};
     \node[] (yoa) [left of=yo] {};
     \node[] (yob) [right of=yo] {};
     \node[state, rectangle,label=right:{$\pi$}] [left of=yoa] (3) {$c_1$};
     \node[state, rectangle,label=left:{$\pi$}] [right of=yob] (4) {$c_2$};

     \node[state,rectangle,label=right:{$\pi@j$}] [below of=3] (6) {$c_1,2,x$};
     \node[] (yo1) [left of=6] {};
     \node[state,rectangle,label=right:{$\pi@i$}] [left of=yo1] (5) {$c_1,2,v_1$};

     \node[draw, rounded corners,label=right:{$\pi@i$}] [below of=5] (10) {$v_1$};
     \node[] (yo2) [right of=10] {};
     \node[draw, rounded corners,label=right:{$\pi@j$}] [right of=yo2] (9) {$x$};

     \node[draw, rectangle,label=right:{$\pi@q$}] [below of=4] (7) {$c_2,1,w_2$};
     \node[] (yo3) [right of=7] {};
     \node[draw, rectangle,label=right:{$\pi@r$}] [right of=yo3] (8) {$c_2,3,w_2$};

     \node[draw, rounded corners,label=right:{$\pi@q$}] [below of=7] (11) {$w_2$};
     \node[] (yo4) [right of=11] {};
     \node[draw, rounded corners,label=right:{$\pi@r$}] [right of=yo4] (12) {$w_2$};

     \path[->] (2) edge node {} (3);
     \path[->] (2) edge node {} (4);


     \path[->] (3) edge [loop left] node {$0$} (3);
     \path[->] (4) edge [loop right] node {$0$} (4);

     \path[->] (6) edge node {} (9);
     \path[->] (5) edge node[left] {} (10);

     \path[->] (7) edge node[left] {} (11);
     \path[->] (8) edge node {} (12);

     \path[->] (3) edge node[left,pos=0.5] {$2i\,\,\,$} (5);
     \path[->] (3) edge node[right,pos=0.5] {$2j+1$} (6);
     \path[->] (4) edge node[left,pos=0.5] {$2q+1$} (7);
     \path[->] (4) edge node[right,pos=0.5] {$\,\,\,\,\,2r$} (8);

  \end{tikzpicture}

\end{center}
\vspace{-5pt}
Here $i=p(v_0 v_1 v_2 v_1)$,
$j=p(v_0 v_1 v_2 x)$,
$q=p(w_0 w_1 w_2)$,
$r=p(w_0 w_1 w_2 w_3 w_2)$ and we assume that the sets
$\pi[i]$ and $\pi[r]$ satisfy $\varphi_S$
and that the sets $\pi[j]$ and $\pi[q]$ do not satisfy $\varphi_S$, leading to the respective
even and odd priorities. In particular, we have $2q+1\neq 2r$ so that $(w_2,\pi@q)$ and
$(w_2,\pi@r)$ are two distinct exit nodes from the certificate $c_2$. Intuitively
they correspond to two different paths through $c_2$ with different $S$-fingerprints;
while in the Emerson-Lei game $C(G)$,
the different fingerprints are dealt with by signalling different
sets of colors, in the parity game $P(C(G))$, the two paths have different effects
on the later-appearence memory $\pi$ and thus lead to different outcome nodes. The self-loops at nodes $c_1$ and $c_2$ correspond to  player $\forall$ giving up, so we assign priority $0$ to these moves.
\end{example}

\subsection{Solving parity games using DAG attraction}

A standard way of solving parity games is by computing a
nested fixpoint of a function that encodes one-step attraction
in the game; the domain of this fixpoint computation then is the set
of all game nodes. For parity games that contain cycle-free parts (DAGs), this process can be improved by instead computing a
nested fixpoint of a function that encodes multi-step attraction along the DAG parts of the game. The domain of the latter
fixpoint computation then does not contain the
internal nodes of the DAG parts, which leads to accelerated fixpoint stabilization. We formalize this idea as follows.

\begin{definition}[DAGs in games]
Let $G=(A,\Omega)$ be a parity game with $A=(V,V_\exists,E)$ and $k+1$ priorities $0,\ldots ,k$. We refer to a set
$W\subseteq V$ of nodes as a \emph{DAG} (directed acyclic graph) if it does not contain an $E$-cycle; then there is no play of $G$ that eventually stays within $W$ forever.
A DAG need not be connected, that is, it may consist of several cycle-free subgames of $G$.
Given a DAG $W\subseteq V$, we write $V'=V\setminus W$ and refer
to the set $V'$ as \emph{real} nodes (with respect to $W$).
A DAG is \emph{positional} if for each existential node $w\in W\cap V_\exists$ in it,
there is exactly one real node $v\in V'$ from which $w$ is reachable without visiting
other real nodes.
\end{definition}
\begin{definition}[DAG attraction]
Given a DAG $W$ and $k+1$ sets $\overline{V}=(V_0,\ldots, V_k)$ of real nodes and a real node $v\in V'$, we say that player $\exists$ can \emph{attract} to $\overline{V}$ from $v$ within $W$
if they have a strategy $\sigma$ such that for all plays $\pi$ that start at $v$ and adhere to $\sigma$, the first real node $v'$ in $\pi$ such that $v'\neq v$ is contained in $V_p$, where $p$ is the maximal priority that is visited
by the part of $\pi$ that leads from $v$ to $v'$.
Given a dag $W$, we define the \emph{dag attractor function}
$\mathsf{DAttr}^W_\exists:\mathcal{P}(V')^k\to \mathcal{P}(V')$ by
$\mathsf{DAttr}^W_\exists(\overline{V})=\{v\in V'\mid
\text{player $\exists$ can attract to $\overline{V}$ from $v$ within $W$}\}$
for $\overline{V}=(V_0,\ldots,V_k)\in \mathcal{P}(V')^k$.
We denote by $t_{\mathsf{Attr}^W_\exists}$ the time
required to compute, for every input $\overline{V}\in\mathcal{P}(V')^k$, the dag attractor of $\overline{V}$ through $W$.

\end{definition}

\begin{remark}\label{rem:safe}
The sets $V_i$ in the above definition correspond to valuations of fixpoint variables in the nested fixpoint computation that is used by the fixpoint acceleration method in Lemma~\ref{lem:pgtrick} below to solve games via DAG attraction. These sets monotonically increase or decrease during the
solution process and at each point of the solution process, a set $V_i$ intuitively holds game nodes for which it currently is assumed that player $\exists$ wins if they can force the game to reach a node from $V_i$ via a partial play in which the maximal priority is $i$.   

During the computation of the DAG attractor to
a tuple $\overline{V}=(V_0,\ldots,V_k)$, we therefore intuitively consider the argument nodes
$\overline{V}$ to be
\emph{safe} in the sense that in order to win from a node $v$, is suffices that existential player has a strategy that ensures that every partial play through the DAG exits it to a node from $V_i$, where $i$ the maximal priority visited by that
play along the DAG. Thus if player $\exists$ can win from all nodes
in $\overline{V}$, then they can win from all nodes in
$\mathsf{DAttr}^W_\exists(\overline{V})$.

In our case, from a node $v$, a strategy $\sigma$ corresponds to
choosing one certificate. Then, player $\forall$ can attract to all
successors of $\forall$-nodes on the certificate.
The path through the certificate to this $\forall$-node and to its
successor outside the certificate shows a certain priority $j$ that
implies the successor must be in the set $V_j$.

\end{remark}

\begin{lemma}[\cite{Hausmann24arxiv}]\label{lem:pgtrick}
Let $G$ be a parity game with priorities $0$ to $k$ and set $V$ of nodes, let $W$ be a positional DAG in $G$, and
let $n=|V|$ and $m=n-|W|$. Then
$G$ can be solved with $\mathcal{O}(m^{\log k+1})$ computations of a DAG attractor;
if $k+1<\log m$, then $G$ can be solved with a number of DAG attractor computations that is polynomial in $m$ (specifically: in $\mathcal{O}(m^5)$).
\end{lemma}
We always have $t_{\mathsf{Attr}^W_\exists}\leq |E|$, using a least fixpoint computation to check for alternating reachability, thereby possibly exploring all DAG edges of the game.
However, in the case that $m<\log n$ and $t_{\mathsf{Attr}^W_\exists}\in \mathcal{O}(\log n)$ (that is, when most of the game nodes are part of a DAG, and DAG attractability can be decided without exploring most of the DAG nodes), Lemma~\ref{lem:pgtrick} enables exponentially faster game solving.

\subsection{Checking certificate existence efficiently}

The parity game $P(C(G))$ that is obtained by applying the LAR construction from Subsection~\ref{subsec:lar} to
the certificate game $C(G)$ has $2^{O (n\cdot\max(d,k)\cdot\log n))}\cdot |S|!$ nodes,
but only $|S|!n$ of these are real nodes: all certificate nodes are internal nodes in a gadget that has a DAG structure.
In this section, we show that DAG attractors in $P(C(G))$ can be computed
efficiently relying on non-emptiness checking of suitable Emerson-Lei automata.

In $P(C(G))$, the DAG attractor to a tuple $\overline{V}=(V_1,\ldots,V_{2|C|+1})$ consists of the nodes
$(v,\pi)$ such that there is a valid certificate starting at $v$ such that
all exits points of the cerficate are safe in the sense of Remark~\ref{rem:safe}, that
is, exiting the certificate with fingerprint $i$ is only possible to nodes $w\in V_i$; we refer to such certificates as \emph{valid and safe}.

Next we show how the existence of valid and safe certificates can efficiently
be checked by using non-emptiness checking for Emerson-Lei automata to find valid and safe
certificate loops that reside over single sets $V_i$ (as the fingerprint does not increase within certificate loops, by the construction of certificates as in the proof of Lemma~\ref{lem:certificateExistence}), and then using a reachability analysis that keeps track of fingerprints to compute safe stems leading (with fingerprint $i$) to some loop over $V_i$.
In more detail, we check for the existence of valid and safe certificates as follows, fixing a permutation $\pi$ and
a tuple $(V_1,\ldots,V_{2|C|+1})$, where each $V_i$ is a set of real nodes in $P(C(G))$.
\begin{itemize}
\item \emph{Make the fingerprints explicit.} Define the graph $M=(V\times [{2|C|+1}],R)$ as follows. Vertices
of the graph are pairs $(v,i)$ consisting of a game node $v\in V$ and a priority $i\in [{2|C|+1}]$, intuitively encoding the largest priority that has been visited since a
DAG has been entered; edges in this graph correspond to game moves but also update the priority value according to visited priorities, that is, $R$ contains exactly the edges
$((v,i),(w,j))$ such that $w\in E(v)$ and $j=\max(i,\Omega_\pi(v,w))$,
where $\Omega_\pi(v,w)$ denotes priority associated to seeing the set $\gamma_C(v,w)$ of
colors on memory $\pi$ (cf. Definition~\ref{defn:LAR}).
\item \emph{Remove unsafe vertices.} A vertex $(u,i)$ such that $(u,\pi@i)$ is contained neither in $V_{2i}$ nor in $V_{2i+1}$ is \emph{unsafe}. Remove from $M$ all vertices $(v,i)$ such that $v\in V_\forall$ and there is $w\in E(v)$ such that $(w,i)$ is unsafe;
these are vertices from where player $\forall$ can access an unsafe exit point of the DAG.
\item \emph{Find safe and valid certificate loops:} Define, for all $i\in[2|C|+1]$, a nondeterministic Emerson-Lei automaton
$A_i=(Q_i,\delta,\alpha_S\land\alpha_W)$ (with singleton alphabet $\{*\}$) by putting
$Q_i=\{(v,i)\mid (v,i)\text{ still exists in }M\}$ and
$\delta((v,i),*)=\{(w,i)\mid w\in E(v)\text{ and } \Omega_\pi(v,w)\leq i \}$
for $(v,i)\in Q_i$.
Compute the non-emptiness region of $A_i$ and call it $N_i$.
\item \emph{Find safe stems:} Remove from $M$ all vertices $(v,0)$ for which there is no $j$ such
that some vertex from $N_j$ is reachable (in $M$) from $(v,0)$.
\end{itemize}
For all vertices $(v,0)$ that are contained in $M$ after this procedure terminates,
there is a safe stem $w$ leading (with maximal priority $j$) to some safe
and valid loop $u$ over $N_j$; $(w,u)$ is a valid certificate for $v$. We state the correctness of the described procedure as follows.

\begin{lemma}\label{lem:attrAut}
Given subsets $\overline{V}=V_0,\ldots V_{2k}$ of the real nodes in $P(C(G))$ and a real node $(v,\pi)$
in $P(C(G))$, we have that player $\exists$ can attract to $\overline{V}$ from $(v,\pi)$ within
$\mathsf{Cert}\times\Pi(S)$ if and only if
the set $M$ contains the pair $(v,0)$ after execution the above procedure
(for parameters $\overline{V}$ and $(v,\pi)$).
\end{lemma}

\begin{corollary}\label{cor:attrComp}
DAG attractors in $P(C(G))$ can be computed in time $\mathcal{O}(d!dt)$ where $t$ denotes the time it takes
to check $\alpha_S\wedge\alpha_W$ automata of size $n$ for non-emptiness.
\end{corollary}
\begin{proof}
In order to compute a DAG attractor in $P(C(G))$, it suffices to execute the above
procedure once for each $\pi\in\Pi(S)$, that is, $d!$ many times; a single execution of the procedure can be implemented in time $\mathcal{O}(dt)$.
\end{proof}
\subsection{Faster solution of obliging games}

We are now ready to state the main result of this section.

\begin{theorem}\label{thm:certSolve}
Certificate games for objectives $\alpha_S$ and $\alpha_W$ with $n$ nodes and $d:=|S|$
colors for the strong objective can be solved
in time $\mathcal{O}((d! n)^{5}d!dt)$, where $t$ denotes the time it takes
to check Emerson-Lei automata of size $n$
and with acceptance condition $\alpha_S\wedge\alpha_W$  for non-emptiness.
If $\alpha_S$ is a parity objective, then the runtime bound
is $\mathcal{O}(n^{\log (2d+1)}dt)$.
\end{theorem}
\begin{proof}
By Lemma~\ref{lem:LAR}, it suffices to solve the paritized version
$P(C(G))$ of $C(G)$.
By Lemma~\ref{lem:attrAut}, computing DAG attractors in $P(C(G))$
can be done in time $\mathcal{O}(d!dt)$.
As we have $d<\log(d!\cdot n)$, $P(C(G))$ can be solved with
$(d!n)^5$ computations of a DAG attractor by Lemma~\ref{lem:pgtrick}.
If $\alpha_S$ is a parity objective, then the LAR construction is not necessary as $C(G)$ already is a parity game; DAG attractors then can be computed in time $\mathcal{O}(dt)$  and
$C(G)$ can be solved with $\mathcal{O}(n^{\log (2d+1)})$ computations of a DAG attractor.
\end{proof}

We collect results on the complexity of non-emptiness checking of Emerson-Lei automata with
acceptance condition $\alpha_S\wedge\alpha_W$ for specific $\alpha_S$ and $\alpha_W$.
It is known (cf. Table 2. in \cite{BaierBD00S19}) that while the problem is in \textsc{P} for most frequently used objectives (subsuming automata with generalized B\"uchi, Rabin or Streett conditions, with linear or quadratic dependence on the number of colors), it
is \textsc{NP}-complete for Emerson-Lei conditions.
For combinations of such objectives, the problem remains in \textsc{P} unless
one of the objectives is of type Emerson-Lei:
\begin{lemma}\label{lem:SANDRemptiness}
The Rabin($d$)$\land$Streett($k$) non-emptiness problem is in P (more precisely: in $\mathcal{O}(mdk^2)$ for automata with $m$ edges).
\end{lemma}
\begin{proof}
Let $A$ be an Emerson-Lei automaton with $n$ nodes, $m$ edges and acceptance condition
Streett$(d)\land$ Rabin($k$). We check $A$ for non-emptiness as follows.
Let the Rabin($k$) condition be encoded by $k$ Rabin pairs $(E_i,F_i)$.
For each $1\leq i\leq k$, check the same automaton
but with acceptance condition Streett$(d)\land \mathsf{Inf}(F_i)\land
\mathsf{Fin}(E_i)$ for emptiness; call this automaton $A_i$.
The acceptance condition of $A_i$ can be treated
as a Streett($d+2$) condition where the two additional Streett pairs
are $(\top,F_i)$ and $(E_i,\bot)$. As a state in $A$ is non-empty if and only
if there is some $i$ such the state is non-empty in $A_i$, it suffices to check
the $k$ many Streett($d+2$) automata for emptiness. The claim follows
from the bound on emptiness checking for Streett automata given in~\cite{BaierBD00S19}.
\end{proof}

Using the equivalence of certificate games
and obliging games (Theorem~\ref{theorem:gameCorrectness}),  Theorem~\ref{thm:certSolve} together with the described complexities of emptiness checking yields improved upper bounds on the runtime complexity of solving obliging games, shown in the table below; we let $n$ denote the number of nodes;
$m$ the number of edges; $b$ the size of a nondeterministic B\"uchi automaton accepting $\alpha_W$; and $t$ the
time required for emptiness checking an Emerson-Lei automaton of size
$n$ with acceptance condition $\alpha_S\wedge\alpha_W$; finally we let
$o$ abbreviate $\max(|W|,|S|)$.
Recall that the approach from~\cite{ChatterjeeHL10}
reduces an obliging game to a standard game in which the  objective $\alpha'$ is the conjunction of $\alpha_S$ with a B\"uchi objective, incurring blowup $b$ on the arena size. This game then can be transformed to a parity game
$\mathcal{G}$ (with parameters $v,e$ and $r$) incurring additional blowup for the LAR transformation of $\alpha'$ to
a parity objective. Notice the definition of $e$ (an underapproximation of the number of edges in
$\mathcal{G}$) in column 4 and its usage in column 5.

For instance for an obliging game with objectives $\alpha_S=$
Streett($d$) and $\alpha_W=$generalized B\"uchi($d$) consisting of the
Streett requests, the approach from~\cite{ChatterjeeHL10} reduces the
game to a parity game $\mathcal{G}$ with $d!nd$ nodes, at least $d!md$
edges, and $2d$ priorities; such a game can be solved
in time $\mathcal{O}(d!md^6(d!n)^5)$. In contrast, emptiness checking for $\alpha_S\wedge\alpha_W$-automata is just emptiness checking for generalized B\"uchi automata so that our new
algorithm solves such games in time $\mathcal{O}(d!md^2(d!n)^5 )$.
For objectives $\alpha_S=$ Rabin($d$) and $\alpha_W=$ Streett($k$), the approach from~\cite{ChatterjeeHL10} has time
complexity $\mathcal{O}(m(d!k2^k)^6n^5)$ while our algorithm has time complexity just $\mathcal{O}(m(d!)^6 n^5do^3)$; this is due to the fact that B\"uchi automata that recognize Streett objectives are of exponential size (as they have to guess a set of Streett pairs and then verify their satisfaction), while emptiness checking for Streett($d$)$\land$ Rabin($k$)-automata can be done in time cubic in $o$.
\begin{table}
\begin{center}
\begin{footnotesize}
\begin{tabular}{| l | l || c | l | l || c|l|}
\hline
type of $\alpha_S$ & type of $\alpha_W$ &
$b$& $|\mathcal{G}|(v,e,r)$~\cite{ChatterjeeHL10}&
\quad time \cite{ChatterjeeHL10} & $t$~\cite{BaierBD00S19} & time here\\
\hline
\hline
\multirow{3}{*}{parity($d$)}
 & Rabin($k$) & $k+1$ & \multirow{3}{*}{$dnb,dmb,d$} &
 \multirow{3}{*}{$\mathcal{O}(e(dnb)^{\log d})$} & $\mathcal{O}(mo^2)$ &
 \multirow{3}{*}{$\mathcal{O}(n^{\log d+1}dt)$}\\\cline{2-3}\cline{6-6}
 & Streett($k$) & $2^kk$ &  & & $\mathcal{O}(mo^2)$ &  \\\cline{2-3}\cline{6-6}
 & EL($k$) & $2^kk$  & & & $\mathcal{O}(mo^2 2^{o})$ & \\
\hline
\hline
Rabin($d$) & Rabin($k$) & $k+1$ & \multirow{3}{*}{$d!nb, d!mb, 2d$} & \multirow{3}{*}{$\mathcal{O}(e(d!nb)^5)$} & $\mathcal{O}(mo^3)$ & \multirow{3}{*}{$\mathcal{O}((d!n)^5 d!dt)$} \\\cline{2-3}\cline{6-6}
or & Streett($k$)& $2^k k$ &  & & $\mathcal{O}(mo^3)$ & \\\cline{2-3}\cline{6-6}
Streett($d$) & EL($k$)& $2^k k$ & & & $\mathcal{O}(mo^2 2^{o})$ & \\
\hline
\hline
EL($d$) & EL($k$) & $2^k k$ & $d!nb, d!mb, 2d$ & $\mathcal{O}(e(d!nb)^5)$ & $\mathcal{O}(mo^2 2^{o})$ & $\mathcal{O}((d!n)^5 d! dt)$ \\
\hline
\hline
Streett($d$) & g.B\"uchi($d$) & $d$ & $d!nb, d!mb, 2d$ & $\mathcal{O}(e(d!nb)^5)$ & $\mathcal{O}(md)$ & $\mathcal{O}((d!n)^5d! dt)$ \\
\hline
\hline
GR[1]($d,k$) & g.B\"uchi($d$) & $d$ & $dknb, dkmb, 3$ & $\mathcal{O}(e(dknb)^2)$ & $\mathcal{O}(md)$ & $\mathcal{O}((dk)^3n^2 dt)$ \\
\hline

\end{tabular}
\end{footnotesize}
\end{center}
\caption{Comparison of runtime complexities for solving obliging games of various types.}
\end{table}
\vspace{-10pt}

\section{Conclusion}
We propose a new angle of looking at the solution of obliging games.
In contrast to previous approaches that have been based on single-step game reasoning, our method
requires players to make promises about their long-term future behavior,
which we formalize using the concept of certificates (or, more generally, witnesses). This new approach to
obliging games enables us to not only show their determinacy (with strategies that contain additional information), but
to also significantly improve previously existing upper bounds both on the size of graciously winning strategies, and on the worst-case runtime complexity of the solution of such games.

Technically, we use our new approach to show that the strategy sizes for Emerson-Lei obliging games with strong objective $\alpha_S$ and
weak objective $\alpha_W$ are linear in the number $|W|$ of colors used in the weak objective; we obtain a similar polynomial dependence on $|W|$ for the runtime of solving obliging games, however with the important exception of the case where $\alpha_W$
is a full Emerson-Lei objective that cannot be expressed by a simpler (e.g. Rabin or Streett) objective.
In previous approaches, those dependencies on $|W|$ have been, in general, exponential.
Thereby we show that the strategy complexity of $\alpha_S$ / $\alpha_W$ obliging games
is not significantly higher than that of standard games with objective just $\alpha_S$, and that in many cases, this holds for the runtime complexity as well.

We leave the existence of canonical certificates (cf. Remark~\ref{rem:canon}) as an open question for future work;
such canonical certificates would allow for the extraction of yet smaller graciously winning strategies for obliging games.

We solve certificate games by using the LAR reduction to obtain equivalent parity games and then solving these parity games by fixpoint acceleration, computing nested fixpoints of a function that checks for certificate existence.
We conjecture that it is possible to directly compute the more involved nested fixpoint associated to Emerson-Lei objectives (as given in~\cite{HausmannLP24}) over the original game arena; this would avoid the LAR reduction step in the solution process. 
\bibliography{lib}

\clearpage
\section{Appendix}

\noindent \textbf{Full proof of Lemma~\ref{lemma:gracious games} (Correctness of witness games):}
\begin{proof}
\begin{itemize}
	\item[$\Rightarrow$]
Let $\sigma:V^*\cdot V_\exists \rightarrow V$ be a graciously winning
strategy for $v$.
Consider the strategy tree $T \subseteq V^+$, where $v\in T$ and for
every $wv'\in A^+$ we have that if $v'\in V_\exists$ then $wv'\sigma(wv')\in
T$ and if $v'\notin V_\exists$ then $wv'v'' \in T$ for every
$(v',v'')\in E$.
By definition of gracious strategies, for every $w\in T$ there is a
play $\pi(w)$ such that $w\pi(w)$ is compatible with $\sigma$ and
satisfies $\alpha_W$.
Furthermore, every play compatible with $\sigma$ (particularly, also
$w\pi(w)$) satisfies $\alpha_S$.
We now define a strategy $\sigma'$ for player $\exists$ in $W(G)$.
Consider a finite play $\pi'$ in $W(G)$ such that $\pi'$ ends in $v'\in
V'_\exists$. Then, $w=\seq_A(\pi')$ is a finite play in $A$.
We put $\sigma'(\pi')=\pi(w)$.
That is, $\sigma'$ moves from $v'$ to the witness $\pi(w)$ appearing
in $T$.

To see that $\sigma'$ is a winning strategy, consider an infinite play $\pi'$ in $W(G)$ that starts at $v$ and is compatible with $\sigma'$.
\begin{itemize}
\item[--]
If $\pi'\in V'^* \cdot (V'_\forall)^\omega$, the let $w$ be the longest
prefix of $\pi'$ ending in a vertex in $V'_\exists$ and let $v'$ be the
last vertex in $w$.
It follows that $\seq_A(\pi')=\seq_A(w)\cdot \sigma'(\seq_A(w)v')$.
By the choice of $\sigma'(\seq_A(w)v')$ we have that $\seq_A(\pi')$
satisfies $\alpha_W$. Furthermore, by $\seq_A(\pi')$ being a play
compatible with $\sigma$ we conclude that $\seq_A(\pi')$ also satisfies
$\alpha_S$. Thus, $\pi'\in \alpha'$.

\item[--] If $\pi'\in V'^*\cdot (V'_\exists \cdot (V'_\forall)^*)^\omega$, then
$\seq_A(\pi')$ is a
play compatible with $\sigma$.
We conclude that $\seq_A(\pi')$ satisfies $\alpha_S$ and hence
$\pi'\in\alpha'$.
\end{itemize}
\item[$\Leftarrow$]
Let $\sigma'$ be a winning strategy in $W(G)$ for $v$.
We construct a graciously winning strategy $\sigma$ in $G$
by using the witnesses.
Alongside $\sigma$ we keep track of all the histories ending also in
$V_\forall$ which have been handled.
Thus, we build $\sigma$ and $T\subseteq A^*$ simultaneously such that
$T$ is prefix closed.
Consider the vertex $v$. By definition, $\sigma'(v)=\pi$ such that
$\pi \in A^\omega$.
For every prefix $w'$ of $\pi$, add $w'$ to $T$. This clearly keeps $T$
prefix closed.
Furthermore,
consider a partition $\pi=w'v'v''w''$ of $\pi$, where $w'\in A^*$,
$v'\in V_\exists$, $v''\in V$, and $w''\in A^\omega$.
Then, define $\sigma(w'v')=v''$.
Notice that in the case that $w'=\epsilon$, we have $v'=v$.

Consider a finite play $wv$ compatible with $\sigma$ such that
$wv\notin T$ and every prefix of $wv$ is contained in $T$.
Then there is a play $\pi_{wv}$ in $W(G)$ such that
$\seq_A(\pi_{wv})=wv$ and $\pi_{wv}$ either ends in vertex
$v$ or can be extended by $v$.
Assume without loss of generality that $\pi_{wv}$ ends in vertex $v$.
Then, $\sigma'(\pi_{wv})=\pi'$ for some $\pi'\in A^\omega$.
As before, for every prefix $w'$ of $\pi'$ add $wvw'$ to $T$, which
again keeps $T$ prefix closed.
Furthermore, consider a partition $w'v'v''w''$ of $\pi'$ such that
$w'\in A^*$, $v'\in V_\exists$, $v''\in V$, and $w''\in A^\omega$.
Then, define $\sigma(wvw'v')=v''$.

We have to show that $\sigma$ is graciously winning:
\begin{itemize}
	\item
Consider an infinite play $\pi$ that is compatible with $\sigma$.
Either $\pi$ is obtained from a finite number of moves of $\sigma'$, in
which case the last move made by $\sigma'$ ensures that $\pi$ satisfies
$\alpha_S$.
Or $\pi$ is obtained from infinitely many moves of $\sigma'$, in which
case $\pi$ satisfies $\alpha_S$ again.
\item
Consider a prefix $p$ of a play $\pi$ that is compatible with $\sigma$.
By definition $p$ is obtained by some prefix $p'$ of $p$ and a
partition of the play $\pi'$ such that $(p',\pi')=\sigma'(p)$.
It follows that $p'\pi'$ extends $p$ and satisfies $\alpha_W$.
\end{itemize}
\end{itemize}
\end{proof}

\noindent \textbf{Full proof of Lemma~\ref{lem:certificateExistence} (Certificate existence):}
\begin{proof}
First we construct a suitable stem, intuitively by taking a finite prefix (of sufficient length) of $\pi$ and removing redundant loops from this prefix.
To this end,
we consider the finite sequence $\rho=(\pi_0,C_0)(\pi_1,C_1)\ldots(\pi_j,C_j)$, where $C_i$ is the
$S$-fingerprint of the finite play $\pi_0\pi_1\ldots \pi_i$ so that
$C_{i}\subseteq C_{i+1}$ for all $i$. The position $j$ is picked to be the least index such
that each node that is visited by $\pi$ from position $j$ on is visited infinitely often by
$\pi$, and such that $C_j$ contains the set of all $S$-colors that are visited by $\pi$ (including also the colors that are visited finitely often by $\pi$). We use this
position to separate the stem and the loop of the prospective certificate
as it is the first position at which the
fingerprint of $\pi$ has reached its full extent and from
which on only infinitely often occuring nodes are visited.

Now we repeatedly pick some two distinct indices $p$ and $q$ such that $(\pi_p,C_p)=(\pi_q,C_q)$ and remove
the subsequence $(\pi_{p+1},C_{p+1})\ldots(\pi_q,C_q)$ from $\rho$; such subsequences correspond
to loops in $V$ that are taken by $\pi$ but along which no new colors are added to the
current fingerprint $C_p=C_q$. As $\rho$ is finite to begin with, this process eventually terminates.
In the remaining sequence, each loop in $V$ adds at least one color from $S$ to
the fingerprint; every two pairs in the sequence are distinct, implying that
its length is bounded by $|S|\cdot |V|$.
We define the stem $w$ to consist of the node components of this shortened sequence.
As $\pi$ is a play on $A$ and $w$ is obtained from a finite prefix of $\pi$ by removing loops,
$w$ is a finite play on $A$ as well. By construction, whenever $w$ visits some node
$v$ with $S$-fingerprint $C$, so does $\pi$.

It remains to construct a suitable loop that starts at the end of $w$. To this end,
let $V_\pi$ and $E_\pi$ denote the set of nodes and moves, respectively,
 that occur infinitely often in $\pi$.
Then $(V_\pi,E_\pi)$ is a strongly connected sub-graph of $(V,E)$, that is, for all
nodes $v,v'\in V_\pi$, there is a finite play of length at most $|V_\pi|$ that starts
at $v$, ends in $v'$ and uses only moves from $E_\pi$. It follows that there is, for all $v\in V_\pi$ and all moves
$e=(v_1,v_2)\in E_\pi$, a finite play of length at most $|V_\pi|+1$ starting at $v$ and containing the move $e$: there is a play $\tau$ of length at most $|V_\pi|$ from $v$ to $v_1$; Thus $\tau v_2$ is a play of length at most $|V_\pi|+1$ that starts at $v$ and contains the move $e$.

Let $C^S_\pi\subseteq S$ and $C^W_\pi\subseteq W$ be the sets of colors
that occur infinitely often in $\gamma_S(\pi)$ and $\gamma_W(\pi)$, respectively.
Recall that $\pi$ satisfies $\varphi_S$ and $\varphi_W$, that is,
that $\gamma_S(\pi)\models \varphi_S$ and $\gamma_W(\pi)\models \varphi_W$.
We construct the loop $u$, visiting one-by-one all the colors
contained in $C^S_\pi\cup C^W_\pi$, writing $C^S_\pi\cup C^W_\pi=\{c_1,\ldots,c_{o}\}$
and noting that $o\leq |S|+|W|$. The definition is inductive:
let $v_0$ be the last node in the stem $w$ constructed above.
For $1<i\leq o+1$, let $v_{i-1}$ be the last node in the play $\tau_{i-1}$ constructed in the step for $i-1$.
Then we define $\tau_i$ to be some play of length at most $|V_\pi|+1$
that starts at $v_{i-1}$ and uses some move with color $c_i$.
We have shown above that such a play always exists. Finally, let $\tau_f$ be some
play of length at most $|V_\pi|$ that starts at $v_{o}$ and ends in node that has
a move to the first node of $\tau_1$.
Then the loop $u=\tau_1 \tau_2\ldots \tau_{o}\tau_f$ is a finite play of length at most
$(|V_\pi|+1)\cdot (|S|+|W|+1)$ that visits exactly the colors contained
in $C^S_\pi\cup C^W_\pi$.
Furthermore, the edge $(\tau_f,\tau_1)$ satisfies
$\{\gamma_S(\tau_f,\tau_1),\gamma_W(\tau_f,\tau_1) \subseteq C^S)^\pi
\cup C^W_\pi$.
As $\pi$ satisfies both $\varphi_S$ and $\varphi_W$, so does the infinite play $u^\omega$, showing
that $(w,u)$ is a certificate for $v$ in $G$.
By construction, whenever $wu^\omega$ visits some node
$v$ with $S$-fingerprint $C$, so does $\pi$.
\end{proof}

\begin{example}\label{ex:cert}
To see how the proof of Lemma~\ref{lem:certificateExistence} extracts certificates from witnesses, consider the game depicted below,
using colors $S=W=\{a,b,c,d\}$ and two objectives over the same colors and
color assignments
\begin{align*}
\varphi_S&=(\mathsf{Inf}~a \to \mathsf{Inf}~c)\wedge \mathsf{Fin}~b
& \varphi_W&=\mathsf{Inf}~d
\end{align*}
\begin{center}
\tikzset{every state/.style={minimum size=15pt}, every node/.style={minimum size=15pt}}
    \begin{tikzpicture}[
		auto,
    node distance=1.8cm,
    semithick
    ]
     \node[state] (0) {$x$};
     \node[state](1) [right of=0] {$y$};
     \node[state](2) [right of=1] {$z$};

     \path[->] (0) edge node[] {$\{a\}$} (1);
     \path[->] (1) edge [bend left] node[] {$\{a\}$} (2);
     \path[->] (2) edge [bend left] node[] {$\{a,d\}$} (1);
     \path[->] (1) edge [loop above] node[] {$\{b\}$} (1);
     \path[->] (2) edge [loop above] node[] {$\{a,c\}$} (2);
  \end{tikzpicture}

\end{center}
Then $\pi=x y y z (y z z)^\omega$ is a play that results in the sequence
$\{a\}\{b\}\{a\}(\{a,d\}\{a\}\{a,c\})^\omega$ of sets of colors.
The sequence satisfies
$\varphi_S$ and $\varphi_W$
since the colors $a$, $c$ and $d$ occur infinitely often in it, but color $b$ does not.
Making the $S$-fingerprints in $\pi$ explicit, we obtain the sequence
\begin{align*}
\hspace{-15pt}\rho=(x,\emptyset)(y,\{a\})(y,\{a,b\})(z,\{a,b\})(y,\{a,b\})(z,\{a,b\})
((z,\{a,b,c\})(y,\{a,b,c\})(z,\{a,b,c\}))^\omega
\end{align*}
We note that after taking the first transition from $x$ to $y$, all nodes that $\pi$ visits
are visited infinitely often by $\pi$. Furthermore, $\pi$ is cyclic from the third
visit of $y$ on. The $S$-fingerprint however reaches its full extent $\{a,b,c\}$
only upon the third visit of $z$, that is, at the end of the first iteration of the loop
$y z z$.

We obtain a stem from the fingerprint-increasing prefix
\begin{align*}
\rho_w=(x,\emptyset)(y,\{a\})(y,\{a,b\})(z,\{a,b\})(y,\{a,b\})(z,\{a,b\})(z,\{a,b,c\})
\end{align*}
of $\rho$ by removing loops that do not change the fingerprint. This leads to the sequence
\begin{align*}
\rho_{w_1}=(x,\emptyset)(y,\{a\})(y,\{a,b\})(z,\{a,b\})(z,\{a,b,c\}).
\end{align*}
The node components of this sequence yield the stem $w=x y y z z$.

To obtain a suitable loop, we consider the
subgraph $(V_\pi,E_\pi)$ that is given by the moves that are taken infinitely often in $\pi$;
this graph consists of the edges $(y,z)$, $(z,z)$ and $(z,y)$; the colors
$a$, $c$ and $d$ are visited infinitely often by $\pi$, so we construct the loop $u=\tau_a\tau_c\tau_d\tau_f$, using only edges from $E_\pi$,
where $\tau_a=z y$, that is, $\tau_a$ is a play that starts at the end node $z$ of the stem $w$ and
takes some edge with color $a$ (in this case we pick the edge $(z,y)$); furthermore, $\tau_c=y z z$ is a play that starts at the end node $y$ of $\tau_a$ and uses the edge $(z,z)$, seeing color $c$.
Then $\tau_d=z y$ is a play that starts at the end node of $\tau_c$ and uses the edge $(z,y)$ with color $d$,
and finally, $\tau_f=y z$ is a play connecting the end node of $\tau_d$ and the start node of $\tau_a$, closing the loop. This results in an overall certificate
\begin{align*}
wu=x y y z z ~y z z y z
\end{align*}
with sequence $\{a\}\{b\}\{a\}\{a,c\}\{a,d\}(\{a\}\{a,c\}\{a,d\}\{a\}\{a,d\})^\omega$ of sets of colors.
Thus the certificate visits the colors $a,c$ and $d$ (but not $b$) infinitely often so that it
satisfies both $\varphi_S$ and $\varphi_W$. Furthermore, for every node
$v$ that is visited with $S$-fingerprint $C$ in the certificate, the
node $v$ is visited with fingerprint $C$ in $\pi$ as well.
In particular the choice of the starting point of the loop ensures that
the $S$-fingerprints for all nodes in the loop have reached the full extent $\{a,b,c\}$.
\end{example}

\begin{lemma}\label{lem:witCert}
Let $G$ be an obliging game and let $v$ be a node in $G$.
If player $\exists$ wins
from $v$ in $W(G)$,
then player $\exists$ wins from $v$ in $C(G)$.
\end{lemma}

\begin{proof}
Let $\sigma$ be a winning strategy for player $\exists$ in $W(G)$.
We inductively construct a strategy $\tau$ for the game $C(G)$; as
invariant of the inductive
construction, we associate with every finite play
$\pi$ of $C(G)$ that adheres to the strategy $\tau$ constructed so far,
a finite play $\rho_\pi$ of $W(G)$ that adheres to $\sigma$.
In the base case of the empty play $\epsilon$ that consists just of
$v$, we put
$\rho_\epsilon=\epsilon$.

For the inductive step, let $\pi w_\exists$ be a finite play ending in
$w_\exists$ that
adheres to the part of $\tau$ that has been constructed so far, and
let $\rho_\pi w$ be the associated play that adheres to $\sigma$.
As $\sigma$ wins $w$, there is a witness $\epsilon$ such that
$\sigma(\rho_\pi)=\epsilon$.
From $\sigma$ being winning, we conclude that $\epsilon$ satisfies
$\gamma_W$ and $\gamma_S$.
By Lemma~\ref{lem:certificateExistence}, there is some certificate
$c=w_0 w_1\ldots
\in \mathsf{Cert}(w)$ such that for all positions $i$ in $c$, there is
a position $j$ in $\epsilon$ such that $w_i=\xi_j$ and the
$S$-fingerprints
of
$w_0\ldots w_i$ and $\xi_0\ldots\xi_j$ coincide.
We put $\tau(\pi w_\exists)=c$. For every position $i$ in $c$ such that
$w_i\in V_\forall$ and every $w'\in E(w_i)$, we have a move $(c,(c,i,w'))$ in $C(G)$.
Then there is some $j$ such that $\xi_j=w_i$ and the $S$-fingerprints
of $w_0\ldots w_i$ and $\xi_0\xi_1\ldots\xi_j$ coincide, where
$w_0=\xi_0=w$.
We extend the play $\rho_\pi w$ that is associated
with $\pi w_\exists$ to the play $\rho w\xi_1 \ldots \xi_jw'$ and
associate it to
the play $\pi w_\exists c w_i w'$. By construction, these extended plays
are compatible
with $\sigma$ and the part of $\tau$ that has been constructed so far.
For moves from $E$
taken in $C(G)$, the associated play is updated accordingly.

It remains to show that $\tau$ is a winning strategy for $v$ in $C(G)$,
that is, that
every play of $C(G)$ that starts at $v$ and adheres to $\tau$
satisfies $\varphi_S$. To this end, let $\pi$ be a play
that starts at $v$ and adheres to $\tau$ and let $\rho_\pi$ be the play
in $W(G)$ that is associated with $\tau$
according to the inductive invariant of the construction of $\tau$;
recall that $\rho_\pi$ adheres to $\sigma$. We note that by construction, player $\exists$ always is able to pick a valid certificate as long as they follow strategy $\tau$. As player $\exists$
graciously wins $v$ using the strategy $\sigma$, $\rho_\pi$ satisfies
$\varphi_S$. By construction, the $S$-fingerprints in $\pi$ and
$\rho_\pi$ coincide
so that $\tau$ satisfies $\varphi_S$ as well.
\end{proof}

\noindent\textbf{Proof of Lemma~\ref{lem:LAR} (Correctness of lazy LAR reduction):}

\begin{proof}
By slight abuse of notation, we let $\Omega(\tau)$ denote
the maximal $\Omega$-priority that is visited in $\tau$, having
$\Omega(\tau)=2(p(\tau))$ if $\pi_0[p(\tau)]\models\varphi_C$
and
$\Omega(\tau)=2(p(\tau))+1$ if $\pi_0[p(\tau)]\not\models\varphi_C$.
\begin{itemize}

	\item[$\Rightarrow$] Let $\sigma=(\Pi(C),\mathsf{update}_\sigma,\mathsf{move}_\sigma)$ be a strategy
with memory $\Pi(C)$ for player $\exists$ in $G$
with which they win every node from their winning region.
It has been shown in a previous LAR reduction for Emerson-Lei games~\cite{HunterD05} that winning strategies
with this amount of memory always exist.
We define a strategy $\rho=(\Pi(C),\mathsf{update}_\rho,\mathsf{move}_\rho)$
with memory $\Pi(C)$ for player $\exists$ in $P(G)$
by putting $\mathsf{update}_\rho(\pi,((v,\pi'),(w,\pi'')))=\mathsf{update}_\sigma(\pi,(v,w))$
and $\mathsf{move}_\rho((v,\pi'),\pi)=(w,\pi@\gamma_C(v,w))$ where
$w=\mathsf{move}_\rho(v,\pi)$. Thus $\rho$ updates the memory and picks moves just as
$\sigma$ does, but also updates the permutation component in $P(G)$ according to the taken moves; hence $\rho$ is a valid strategy.

We show that $\rho$ wins a node $(v,\pi)$ in $P(G)$ whenever $v$ is in the winning region
of player $\exists$ in $G$.
To this end, let $\tau=(v_0,\pi_0)(v_1,\pi_1)\ldots$ be a play of $P(G)$ that starts at
$(v_0,\pi_0)=(v,\pi)$ and is compatible with $\rho$. By construction, $\pi=v_0 v_1\ldots$ is a
play that is compatible with $\sigma$. Since $\sigma$ is a winning strategy for player
$\exists$, we have $\gamma_C(\pi)\models \varphi_C$. There is a number $i$
such that all colors that appear in $\pi$ from position $i$ on occur infinitely often.
Let $p$ be the number of colors that appear infinitely often in $\pi$.
It follows by definition of $\pi@D$ for $D\subseteq C$ (which moves the single
right-most element of $\pi$ that is contained in $D$ to the very front of $\pi$),
that there is a position $j\geq i$ such that the left-most $p$ elements of $\pi_j$
are exactly the colors occuring infinitely often in $\pi$ (and all colors to
the right of $\pi_j(p)$ are never visited from position $j$ on). It follows
that from position $j$ on, $\tau$ never visits a priority larger than $2p$.
To see that $\tau$ infinitely often visits priority $2p$
we note that $\pi_j'[p]\models\varphi_C$
for every $j'>j$, so it suffices to show that $p$ infinitely
often is the rightmost position in the permutation component of $\tau$
that is visited. This is the case since, from position $j$ on,
the $p$-th element in the permutation component
of $\tau$ cycles fairly through all colors that are visited infinitely often by $\pi$.

	\item[$\Leftarrow$]

Let $\rho$ be a positional strategy for player $\exists$ in $P(G)$
with which they win every node from their winning region.
We define a strategy $\sigma=(\Pi(C),\mathsf{update}_\sigma,\mathsf{move}_\sigma)$
with memory $\Pi(C)$ for player $\exists$ in $G$
by putting $\mathsf{update}_\sigma(\pi,(v,w))=\pi@\gamma_C(v,w)$
and $\mathsf{move}_\sigma(v,\pi)=w$ where
$w$ is such that $\rho(v,\pi)=(w,\pi@\gamma_C(v,w))$. Thus $\sigma$ updates the memory and
picks moves just as plays that follow $\rho$ do.

We show that $\sigma$ wins a node $v$ in $G$ whenever $(v,\pi)$ is in the winning region
of player $\exists$ in $P(G)$.
To this end, let $\pi=v_0 v_1\ldots$ be a play of $G$ that starts at
$v_0=v$ and is compatible with $\sigma$. By construction, $\pi$
induces a play $\tau=(v_0,\pi_0) (v_1,\pi_1)\ldots$ with $(v_0,\pi_0)=(v,\pi)$
and $\pi_{i+1}=\pi_{i}@\gamma_C(v_i,v_{i+1})$ for $i\geq 0$
that is compatible with $\rho$. Since $\rho$ is a winning strategy for player
$\exists$, the maximal priority in it is even (say $2p$).
Again, $p$ is the position such that the left-most $p$ elements in the permutation
component of $\rho$ from some point on contain exactly the colors that are
visited infinitely often by $\pi$. It follows that $\gamma_C(\pi)\models\varphi_C$.

\end{itemize}

\end{proof}

\noindent\textbf{Proof of Lemma~\ref{lem:attrAut} (Correctness of efficient DAG attractor computation):}

\begin{proof}
For one direction, let player $\exists$ be able to attract to $\overline{V}$ from
$(v,\pi)$ within $\mathsf{Cert}\times\Pi(S)$. Then there is
a valid certificate $c=v_0\ldots v_m\in \mathsf{Cert}(v)$ for $v$ (with $v_0=v$) such that
for all positions $i$ in $c$ such that $v_i\in V_\forall$
and for all $w\in E(v_i)$, we have
$w\in V_{2p+1}$ or $w\in V_{2p}$, that is, $c$ is safe.
Here,
$p$ is the rightmost position in $\pi$ such that $\pi(p)\in\gamma_S(v_0\ldots v_iw)$. 
Let $u=v_0\ldots v_q$ be the stem of $c$ and $w=v_{q+1}\ldots v_m$ the loop.
Let $i$ be the priority corresponding to the $S$-fingerprint of the stem $u$ with
respect to $\pi$.
Define a run $\tau$ of the automaton $A_i$ from the procedure that computes the set $M$ by
$\tau(0)=(v_{q+1},i)$ and $\tau(j)=(v_{q+1+j},i)$ for $j>0$; here we use $v_{q+1+j}$
to refer to the $j$-th position in the infinite run $w^\omega$. 
As $c$ is a safe certificate, $\tau$ only visits such game nodes $v_r\in V_\forall$ such
that for all successors $w'\in E(v_r)$ of $v_r$, we have $(w',i)\in V_i$.
Thus $\tau$ indeed is an admissible run of $A_i$.
As $c$ is a valid certificate, we have $\gamma_S(u)\models\varphi_S$
and $\gamma_W(u)\models\varphi_W$ so that $\mathsf{Inf}(\gamma_C(\tau))\models \varphi_C$, showing that
$\tau$ is accepting. Thus $(v_{q+1},i)$ is contained in $N_i$.
Now we argue that $(v,0)$ is contained in the graph $M$ after the construction of $M$ terminates.
This indeed is the case since the stem $u$ of the certificate $c$ is safe by assumption; thus the stem corresponds to a path from $(v,0)$ to $(v_{q+1},i)$ that visits
only such nodes $(v_j,p)$ with $v_j\in V_\forall$ where we again have that for all $w'\in E(v_j)$,
$(w',p')\in V_{2p'}$ or $(w',p')\in V_{2p'+1}$, where $p'=\max(p,\Omega_\pi(v_j,w')$.

For the converse direction, let $(v,0)$ be contained in the graph $M$ after the
computation terminates.
Then there is $i$ such that the automaton $A_i$ is non-empty
and there is a path from $(v,0)$ to some $(v_{q+1},i)$ in $M$.
Let $\tau=(v_0,p_0)(v_1,p_1)\ldots$ be a path in $M$ (starting at $(v,0)$ so that $v_0=v$) that witnesses these facts by combining the path from
$(v,0)$ to $(v_{q+1},i)$ and then looping within $N_i$.
We extract a valid certificate from $\tau$ as follows. Define the stem
to be $w=v_0v1\ldots v_{q+1}$. Extract the loop $u$ as in
the proof of Lemma~\ref{lem:certificateExistence}, that is, by
walking through the set of game nodes that occur infinitely often in $\tau$
and visiting all colors from $\gamma_S(u)\cup\gamma_W(u)$.
As $\tau$ is accepting, $\gamma_S(u)\cup\gamma_W(u)\models\varphi_S\wedge\varphi_W$,
showing that $c=wu$ is a valid certificate.
It remains to show that
for all positions $i$ in $c$ such that $v_i\in V_\forall$
and for all $w\in E(v_i)$ such that $w\neq v_{i+1}$, we have
$w\in V_{2p+1}$ or $w\in V_{2p}$,
where $p$ again is the rightmost position in $\pi$ such that
$\pi(p)\in\gamma_S(v_0\ldots v_iw)$.
This is the case since all unsafe vertices have been removed from $M$ so
that $\tau$ visits only safe vertices (that have the required property).
\end{proof}

%

\end{document}